\documentclass[journal]{IEEEtran}

\ifCLASSINFOpdf
 
\else
 
\fi

\usepackage{cite}
\usepackage{amsmath}
\usepackage{amssymb}
\usepackage{amsthm} 
\usepackage{caption}
\usepackage{subcaption}
\usepackage{epstopdf}
\usepackage{booktabs}
\usepackage[table,xcdraw]{xcolor}
\usepackage{graphicx}
\usepackage{float}
\usepackage{color}
\usepackage{latexsym}
\usepackage{algorithm}
\usepackage{multicol,lipsum}
\usepackage[nodisplayskipstretch]{setspace}
\usepackage{optidef}
\usepackage[noend]{algpseudocode}
\usepackage{amsmath}
\usepackage{tikz}
\usepackage{mathdots}
\usepackage{yhmath}
\usepackage{cancel}
\usepackage{color}
\usepackage{siunitx}
\usepackage{array}
\usepackage{multirow}
\usepackage{amssymb}
\usepackage{gensymb}
\usepackage{tabularx}
\usepackage{booktabs}
\usetikzlibrary{fadings}
\usetikzlibrary{patterns}
\usetikzlibrary{shadows.blur}
\usetikzlibrary{shapes}
\usepackage{lipsum}
\usepackage{mathtools}
\usepackage{cuted}
\usepackage{algorithm}
\usepackage{algpseudocode}
\usepackage{pifont}

\newtheorem{theorem}{Theorem}
\newtheorem{remark}{Remark}

\newtheorem{lemma}{Lemma}

\begin{document}
	
%
\title{Percentile Optimization in Wireless Networks—\\Part II: Beamforming for Cell-Edge Throughput Maximization}

%
%
%

\author{Ahmad~Ali~Khan,~\IEEEmembership{Student Member,~IEEE,}
        Raviraj~S.~Adve,~\IEEEmembership{Fellow,~IEEE.}\vspace{-2.6em}
\thanks{Raviraj Adve is with the Edward S. Rogers Sr. Department
	of Electrical and Computer Engineering, University of Toronto, Toronto,
	Ontario, M5S 3G4, e-mail: rsadve@ece.utoronto.ca. Ahmad Ali Khan is with Ericsson R\&D, Ottawa, Ontario, K2K 2V6, email: ahmad.a.khan@ericsson.com. The authors would like to acknowledge the support of Ericsson Canada.}}

\markboth{IEEE TRANSACTIONS ON SIGNAL PROCESSING (DRAFT)}%
{Khan and Adve \MakeLowercase{\textit{et al.}}: Percentile Optimization in Wireless Networks-- Part 2.}
%



\maketitle
{\addvspace{-9.5\baselineskip}}

\begin{abstract}
Part I of this two-part paper focused on the formulation of percentile problems, complexity analysis, and development of power control algorithms via the quadratic fractional transform (QFT) and logarithmic fractional transform (LFT) for sum-least-$q^\mathrm{th}$-percentile (SLqP) rate maximization problems. In this second part, we first tackle the significantly more challenging problems of optimizing SLqP rate via beamforming in a multiuser, multiple-input multiple-output (MU-MIMO) network to maximize cell-edge throughput. To this end, we first propose an adaptation of the QFT algorithm presented in Part I that enables optimization of the complex-valued multidimensional beamforming weights for the SLqP rate utility function. We also introduce a new class of problems which we term as sum-greatest-$q^\mathrm{th}$-percentile weighted mean squared error (SGqP-WMSE) minimization. We show that this class subsumes the well-known sum-weighted mean squared error (WMMSE) minimization and max-WMSE minimization problems. We demonstrate an equivalence between this class of problems and the SLqP rate maximization problems, and show that this correspondence can be exploited to obtain stationary-point solutions for the aforementioned beamforming problem. Next, we develop extensions for the QFT and LFT algorithms from Part I to optimize ergodic \textit{long-term average} or ergodic SLqP utility. Finally, we also consider related problems which can be solved using the proposed techniques, including hybrid utility functions targeting optimization at specific subsets of users within cellular networks.

\end{abstract}


%
\IEEEpeerreviewmaketitle

\section{Introduction}
%
%
%
%
\subsection{Background and Overview}

\IEEEPARstart{I}{n} Part I of this paper, we formulated percentile problems and developed power control strategies for the family of sum-least-$q^{\mathrm{th}}$-percentile (SLqP) rate optimization problems. Specifically, we showed that optimizing throughput for lower percentiles allows us to quantitatively target improvements for cell-edge users in a physically interpretable yet mathematically rigorous fashion. Moreover, by controlling the choice of percentile, we can now control the fundamental tradeoff between favoring cell-centre and cell-edge users in an interference-limited cellular network. At the extremes, SLqP rate maximization problems subsume the well-known sum-rate and max-min-rate problems.

In this second part, we first focus on the more general setting of optimizing beamforming for percentile-rate objectives in MU-MIMO networks. Similar to the power control setting considered in Part I, these problems, with the exception of sum-rate and max-min-rate, have never been directly tackled in the prior literature. At the same time, this setting is of immense practical interest owing to its relevance to multiantenna wireless networks \cite{soret_interference_2018, 3gpp}. In particular, sixth-generation (6G) cellular networks are envisioned to have cell-edge rates that are significantly greater \cite{ziegler_6g_2020} than those possible with current fifth-generation (5G) networks; one of the key technologies enabling this advance is an increased number of transmit and receive antennas. As a result, numerous prior academic works have focused on trying to improve cell-edge rates with a particular emphasis on lower-percentile rates \cite{khan_optimizing_2020,shen_fractional_2018-1,hosseini_optimizing_2018,yenihayat_downlink_2020,garcia-morales_statistical_2018}. Indeed, we note that numerous works have attempted to indirectly optimize cell-edge throughput for cooperative MIMO \cite{do_improving_2018, yang_cell-edge-aware_2017}, massive MIMO \cite{zhang2021improving, ma_interference-alignment_2018} and cell-free or heterogeneous MIMO networks \cite{evangelista_fairness_2019,mankar_load-aware_2018}.

In the context of beamforming, there are two main approaches to improve lower-percentile throughput adopted by prior research works. The first is to utilize the framework of stochastic analysis. With this approach, a beamforming and user selection scheme is chosen beforehand. Using tools from probability theory and stochastic geometry, closed-form expressions for user rates are derived, and the parameters which improve throughput at a desired percentile are determined empirically. Numerous prior works have adopted this approach to improve $10^\mathrm{th}$-percentile rate \cite{hosseini_optimizing_2018,zhu_stochastic_2018} and $5^\mathrm{th}$-percentile rate \cite{yang_cell-edge-aware_2017}. A more general framework for stochastically optimizing cell-edge rates was also presented in \cite{nguyen_coverage_2018}. However, such a strategy is clearly sub-optimal since the beamforming and user selection choices are pre-defined. Furthermore, the process for determining the resource allocation parameters which give the performance at any given percentile is \textit{ad-hoc}.

The second strategy is to utilize the framework of optimization theory. Here, just as with the power control setting, no prior works have directly tackled functional optimization of throughput at any desired percentile via beamforming. Instead, a proxy objective function, usually proportional fair (PF) weighted sum-rate (WSR), is chosen, and optimized. In particular, prior works applying fractional programming have utilized this technique to boost $10^\mathrm{th}$-percentile rate \cite{khan_optimizing_2020,shen_fractional_2018-1}. Other works adopting beamforming techniques to indirectly optimize cell-edge rates include \cite{li_joint_2021,rahman_enhancing_2010}. We emphasize, however, that WSR maximization is, at best, a heuristic for improving lower-percentile rates.

It is also important to note that percentile beamforming optimization problems, even for the special max-min and sum-rate cases, are fundamentally more challenging than power control problems. For example, the max-min-rate beamforming optimization problem is non-convex and NP-hard \cite{razaviyayn_linear_2013}; in contrast, for the power control scenario, it is a quasiconvex problem \cite{zhi-quan_luo_dynamic_2008}, and thus efficiently solvable. Furthermore, the algorithms to deal with non-smooth beamforming problems, such as the max-min optimization problem, usually rely upon smoothing the problem through epigraph formulations \cite{scutari_parallel_2017_1, scutari_parallel_2017,shen_fractional_2018-1,razaviyayn_linear_2013}. As we have seen from Part I, however, this technique inevitably results in a combinatorial number of non-convex constraints, making it computationally infeasible. On the other hand, dealing directly with the non-smooth and non-convex percentile objective function is also undesirable.

A further challenge comes in the form of optimizing \textit{long-term average utility}. Instead of maximizing the utility for a single-time slot, usually a more desirable strategy is to focus on optimizing an exponentially weighted average of the long-term utility; this approach ensures that the fast-fading channel components do not adversely impact resource allocation. For conventional utility functions, such an optimization is typically enabled by decomposing the long-term problem into a sequence of tractable per-timeslot approximations. In the case of proportional fair utility, for example, the logarithm utility function is replaced with a first-order expansion \cite{hossain_adaptive_2011}. This leads to the well-known weighted sum-rate (WSR) maximization approach, in which the weights are successively updated after each time slot. Numerous works have been dedicated to solving the WSR maximization problem in various scenarios \cite{khan_optimizing_2020,shen_fractional_2018-1,shi_iteratively_2011}. However, as discussed in Part I, percentile utility functions are \textit{non-smooth} and are computed \textit{jointly} across all users, making the aforementioned decomposition techniques inapplicable and the optimization more challenging. 

Finally, we note that different groups of cell-edge users may have very different throughput needs and requirements in cellular networks. Accordingly, operators may wish to prioritize different objective functions for different sets of users within a cellular network, or combine multiple utility criteria into a single \textit{hybrid} utility function. For example, the work in \cite{scutari_parallel_2017} examines the multicast beamforming problem achieving max-min fairness across multiple cells; \cite{naghsh_maxmin_2019} examines the weighted max-min beamforming optimization problem. Similarly, the work \cite{naderializadeh_resource_2021} proposes a multiagent deep reinforcement learning approach to maximize a weighted sum of the $100^\mathrm{th}$-percentile throughput (i.e., sum-rate) and $5^\mathrm{th}$-percentile throughput to achieve a desired fairness level. 

Such resource allocation problems are highly challenging to solve owing to their non-convexity and, in the case of percentile utility functions, inherent combinatorial structure. Classic optimization-theoretic resource allocation techniques like fractional programming (FP) and WMMSE can be extended beyond weighted sum-rate; however, the utility function being optimized has to satisfy additional stringent conditions \cite{shen_fractional_2018-1} such as strict convexity and sum-decomposability. However, as noted in Part I, the SLqP utility function, for example, is neither strictly concave nor smooth and cannot be decomposed into a sum of per-user utilities. This renders the prior algorithms unsuitable for optimization of the class of hybrid utility functions.

\subsection{Contributions of Part II}
The contributions of Part II can be briefly summarized as follows:

\begin{itemize}
	
	\item{\textbf{Beamforming via QFT extension:} We develop an extension for the QFT algorithm described in Part I to solve SLqP rate optimization problems for multiuser, multiantenna MIMO networks through beamforming. As with the power control setting, this enables us to focus on cell-edge users by optimizing lower-percentile throughput.}
	
	\item{\textbf{Equivalence of SGqP-WMSE minimization and sum-percentile rate maximization problems:} We introduce a new class of optimization problems, called sum-greatest-$q^\mathrm{th}$-percentile weighted mean squared error (SGqP-WMSE) minimization which subsumes the well-known sum-WMSE (WMMSE) and max-WMSE minimization problems from prior works. Crucially, we prove that solving the SGqP-WMSE minimization problem for a given choice of percentile via a cyclic optimization approach leads to an effective solution of the SLqP rate beamforming problem. This allows us to utilize a powerful unified algorithmic strategy to optimize \textit{all} SLqP rate problems, while also greatly simplifying convergence analysis compared to prior works (e.g. for special cases like the max-min-rate beamforming problem in \cite{razaviyayn_linear_2013} and \cite{scutari_parallel_2017_1}).}

	\item{\textbf{Hybrid Utility Functions:} By employing the calculus of convex function composition, we demonstrate that the proposed algorithmic framework can be extended to optimize \textit{any} concave and non-decreasing utility function. In particular, this includes as special cases well-known non-smooth optimization problems, such as the multitone max-min sum-rate optimization problem in \cite{zhi-quan_luo_dynamic_2008} and the aforementioned weighted max-min-rate optimization in \cite{scutari_parallel_2017_1,scutari_parallel_2017}. We also show how to combine conventional utility functions (e.g., PF utility) with percentile utility functions in order to guarantee fairness for all users.}
	
	\item{\textbf{Long-term average sum-percentile throughput optimization:} Applying standard long-term average utility decomposition methods is not possible for general sum-least percentile utility problems owing to the non-smoothness of the utility functions. Instead, we decompose the problem of long-term average sum-least-percentile rate optimization into a sequence of tractable optimization problems in each time slot. We then apply the proposed optimization algorithms to these sub-problems to compute effective solutions.}
\end{itemize}

\vspace{-1.50em}

\subsection{Notation}
The notation utilized for Part II is largely identical to that for Part I, with a few minor additions and changes. We indicate the set of complex numbers by $\mathbb{C}$. The conjugate of a complex scalar $z$ is denoted by $z^{*}$, while the conjugate transpose of a complex vector $\mathbf{z}$ and complex matrix $\mathbf{Z}$ are denoted by $\mathbf{z}^{\dagger}$ and $\mathbf{Z}^{\dagger}$ respectively. The vectorization of a matrix $\mathbf{X}$ is denoted by $\mathrm{vec}(\mathbf{X})$.

\subsection{Organization}
Part II is organized as follows:
\begin{itemize}
	\item In Section \ref{MQFT}, we formulate the percentile throughput beamforming optimization problem and propose an extension of the QFT algorithm.
	\item We then introduce the class of SGqP-WMSE problems and elucidate their connection with the SLqP-rate problems in Section \ref{SGqPsection}. 
	\item Next, we consider the construction and optimization of hybrid utility functions in Section \ref{hybrids}. 
	\item This is followed by ergodic long-term utility maximization in Section \ref{long_term} and conclusions in Section \ref{conklusion}.
	\item As with Part I, all sections follow an identical order: we first formulate the problem, then describe the proposed approach(es) and finally present numerical results at the end of section.
\end{itemize}
Readers who have no prior context and are interested in mathematical definitions and properties of the percentile utility functions considered in this work are referred to Part I for these details. Part I also focuses on short-term (i.e., single time slot) power control for a multiuser single-input, single-output (SISO) network for both the parallel Gaussian channel and interference-limited scenarios. It also establishes the complexity status of percentile throughput optimization in the latter setting.   

\section{Beamforming for Sum-Percentile Rate Optimization}\label{MQFT}
\subsection{Problem Formulation}\label{QFT_BF_model_formulation}
\small
\begin{figure*}[!t]
	\footnotesize
	\hrulefill
	\begin{equation}
	r_{k,b}\left(\mathbf{V}\right)=\mathrm{log}\left(1+\mathbf{v}_{k,b}^{\dagger}\mathbf{H}_{b\rightarrow k,b}^{\dagger}\left(\sigma^{2}\mathbf{I}+\sum_{\left(b',k'\right)\neq\left(b,k\right)}\mathbf{H}_{b'\rightarrow k,b}\mathbf{v}_{k',b'}\mathbf{v}_{k',b'}^{\dagger}\mathbf{H}_{b'\rightarrow k,b}^{\dagger}\right)^{-1}\mathbf{H}_{b\rightarrow k,b}\mathbf{v}_{k,b}\right)\label{data_rate_beamforming}
	\end{equation}
	\begin{equation}
	\hat{r}_{k,b}\left(\boldsymbol{X},\mathbf{V}\right)=\mathrm{log}\left(1+2\mathrm{Re}\left\{ \mathbf{\boldsymbol{\chi}}_{k,b}^{\dagger}\mathbf{H}_{b\rightarrow k,b}\mathbf{v}_{k,b}\right\} -\mathbf{\boldsymbol{\chi}}_{k,b}^{\dagger}\left(\sigma^{2}\mathbf{I}+\sum_{\left(k',b'\right)\neq\left(k,b\right)}\mathbf{H}_{b'\rightarrow k,b}\mathbf{v}_{k',b'}\mathbf{v}_{k',b'}^{\dagger}\mathbf{H}_{b'\rightarrow k,b}^{\dagger}\right)^{-1}\mathbf{\boldsymbol{\chi}}_{k,b}\right)\label{transformed_rates_QFT}
\end{equation}	
	\hrulefill
\end{figure*}
\normalsize
Consider the downlink of a network comprising $B$ cells, each with a single BS serving $K_{BS}$ users within its cell; thus, the total number of users in the network is $K=BK_{BS}$. We assume that BS $b$ is equipped with $M_b>1$ transmit antennas, while user $k$ being served by this BS is equipped with $N_{k,b}$ receive antennas. Accordingly, we denote the downlink channel from BS $b'$ to the aforementioned user as $\mathbf{H}_{b'\rightarrow k,b}\in\mathbb{C}^{N_{k_{b}}\times M_{b}}$. Thus, $b'=b$ denotes the information-bearing channel from BS $b$ to the $k^\mathrm{th}$ user within its cell; conversely, $b'\neq{b}$ denotes an interference channel. For simplicity, we assume that each user is served using only one data stream; we denote the transmit beamforming weight from BS $b$ to user $k$ by $\mathbf{v}_{k,b}\in\mathbb{C}^{M_{b}\times{1}}$. To avoid notational clutter, we collect all the beamforming variables in the matrix $\mathbf{V}$, the composition of which is given by:
Then the composition of the matrix $\mathbf{V}$, which is used to collect the beamforming weights, is given by
\[
\mathbf{V}\triangleq\left[\mathbf{v}_{1,1}\mathbf{v}_{2,1},\ldots,\mathbf{v}_{K_{B},B}\right]
\].
The achieved rate for the user in question is then given by (\ref{data_rate_beamforming}), where $\sigma^{2}$ represents the user's receiver noise power. 

We collect the user rates in the $K_{BS}\times{B}$ rates matrix $\mathbf{R}\left(\mathbf{V}\right)$, i.e.,
\[
\mathbf{R}\left(\mathbf{V}\right)_{k,b}=r_{k,b}\left(\mathbf{V}\right)
\]
It follows that the problem of optimizing the network-wide SLqP rate is then given by:
\begin{subequations}\label{SPR_problem_BF}
	\begin{align}
	\underset{\mathbf{V^{\mathit{}}}}{\mathrm{maximize}}\quad f_{K_{q}}\left(\mathrm{vec}\left(\mathbf{R}\left(\mathbf{V}\right)\right)\right)\hspace{8.95em}\label{SPR_obj_BF}\\
	\mathrm{subject\,to}\quad\sum_{k=1}^{K_{BS}}\left\Vert \mathbf{v}_{k,b}\right\Vert _{2}^{2}\leq P_{\mathrm{max}};\quad b=1,\ldots,B\label{SPR_constraint_BF}\hspace{1.00em}
	\end{align}
\end{subequations} 
where the constraint in (\ref{SPR_constraint_BF}) ensures that the total transmitted power at each BS is below a maximum threshold of $P_{\mathrm{max}}$. We focus on SLqP rate optimization rather than least-$q^{\mathrm{th}}$-percentile (LqP) rate (i.e., the $K_{q}^{\mathrm{th}}$ smallest rate) owing to the latter's undesirable property in the context of cell-edge throughput maximization. As in Part I, the choice of percentile parameter $q$ allows us to control the tradeoff between favouring cell-edge and cell-centre users. For example, selecting $q=5$ would enable us to optimize $5^{\mathrm{th}}$-percentile throughput in accordance with the targets for next-generation wireless cellular networks.
\begin{remark}
	The percentile program in (\ref{SPR_problem_BF}) is equivalent to maximization of the network sum-rate when $K_q=K$ and maximization of the minimum rate when $K_q=1$.
\end{remark}
The above remark follows directly from Property 4 of the SGqP and SLqP utility functions established in Part I.
\begin{remark}
	Problem (\ref{SPR_constraint_BF}) is non-convex for all values of $K_{q}$, and non-smooth for $K_q\neq{K}$ (i.e., $q=100$).
\end{remark}
The non-smoothness of Problem (\ref{SPR_problem_BF}) follows directly from Property 5 of the SLqP and SGqP utility functions established in Part 1. For the non-convexity, we note that for $K_q\neq{1}$, the sum-of-rates function is non-convex in the beamforming variables in general. Since the SLqP utility function is a pointwise minimum over the \textit{all} sum-rates achieved by every possible subset of $K_q$ users, it follows that the objective in (\ref{SPR_obj_BF}) is also non-convex in the beamforming variables.

For $K_q=1$, we note that the classic max-min-rate beamforming problem can be written in (smooth) hypograph form by introducing a slack variable $t$ as follows:
\begin{subequations}\label{maxmin_problem_BF}
	\begin{align}
	\underset{t,\mathbf{V}}{\mathrm{maximize}}\quad t\hspace{16.400em}\label{maxmin_obj}\\
	\mathrm{subject\,to}\quad\sum_{k=1}^{K_{BS}}\left\Vert \mathbf{v}_{k,b}\right\Vert _{2}^{2}\leq P_{\mathrm{max}};\quad b=1,\ldots,B\label{maxmin_constraint_BF}\hspace{1.30em}\\
	r_{k,b}\left(\mathbf{V}\right)\geq t;\quad\begin{array}{c}
	b=1,\ldots,B\\
	k=1,\ldots,K_{BS}
	\end{array}\label{epigraph_rate_maxmin}\hspace{0.150em}
	\end{align}
\end{subequations} 

The user rates $r_{k,b}\left(\mathbf{V}\right)$ are no longer quasi-concave in the beamforming variables; hence, the $t$-superlevel set in (\ref{epigraph_rate_maxmin}) is non-convex in general, rendering Problem (\ref{maxmin_problem_BF}) non-convex overall when either the number of transmit antennas per BS or receive antennas per user is more than 1. 

The non-convexity and non-smoothness of Problem (\ref{SPR_problem_BF}) make optimization extremely challenging. On the one hand, traditional gradient-based algorithms are not applicable owing to the non-smoothness of the percentile utility functions. On the other hand, smoothing the problem via hypograph reformulation leads to a combinatorial number of non-convex constraints when $K_q\neq{1}$, rendering this approach unsuitable for all SLqP rate maximization problems other than the max-min-rate optimization problem. 

Instead, we develop two distinct minorization-maximization approaches to tackle this class of problems. Similar to the power control problem, these approaches introduce auxiliary variables that transform the original non-convex rate expressions into equivalent block-concave functions. By iteratively optimizing one block of variables while holding the others fixed, the proposed algorithms then converge to directional stationary points of the original objective function. 

\noindent\textit{Differences from power control:} A natural starting point for this is to extend the QFT approach developed in Part I for the beamforming setting. The LFT approach cannot, however, be straightforwardly adapted for the multidimensional setting (since it does not result in an equivalent objective function that is block-concave in the power variables). Furthermore, the expressions for signal and interference power are now quadratic functions of the optimization variables (unlike the power control setting, in which they are linear), thereby making approaches like successive convex approximation even more difficult to implement. This prompts us to adapt the well-known WMMSE approach for this class of problems.

\vspace{-0em}
\subsection{Multidimensional Quadratic Fractional Transform (MQFT)}\label{MQFTmethod}
We begin by re-stating the multidimensional complex-valued extension of the quadratic fractional transform, first proposed by Benson \cite{benson_global_2004} and detailed by Shen and Yu \cite{shen_fractional_2018}: 

\begin{lemma}\label{quadratic_transform}
	Let $\mathbf{a}(\mathbf{V}):\mathrm{\mathbb{C}}^{d_{1}\times{1}}\mapsto\mathbb{C}^{d_{2}\times{1}}$, $\mathbf{B}(\mathbf{V}):\mathrm{\mathbb{C}}^{d_{1}\times{1}}\mapsto\mathbb{S}_{++}^{d_{2}\times d_{2}}$. Then we have
	\begin{equation}
	\mathbf{a^{\dagger}}(\mathbf{V})\mathbf{B}^{-1}(\mathbf{V})\mathbf{a}(\mathbf{V})=\underset{\boldsymbol{\chi}}{\mathrm{max}}\quad2\mathrm{Re}\left\{ \mathbf{\boldsymbol{\chi}}^{\dagger}\mathbf{a}\left(\mathbf{V}\right)\right\} -\mathbf{\boldsymbol{\chi}}^{\dagger}\mathbf{B}(\mathbf{V})\mathbf{\boldsymbol{\chi}}\label{QFT_transform}
	\end{equation}
	where $\boldsymbol{\chi}\in\mathbb{C}^{d_{2}\times1}$ is an auxiliary variable.
\end{lemma}
\begin{proof}
	The expression on the right-hand side of (\ref{QFT_transform}) is concave in $\boldsymbol{\chi}$ and can be manipulated by completing the square as in \cite{shen_fractional_2018}. It then becomes clear that the expression is maximized when the auxiliary variable $\boldsymbol{\chi}$ is chosen as:
	\begin{equation}\label{optimal_quadratic_variables}
	\mathbf{\boldsymbol{\chi}^{\mathrm{opt}}=B}^{-1}(\mathbf{V})\mathbf{a^{\dagger}}(\mathbf{V})
	\end{equation}
	Substituting this back into the right-hand side of (\ref{QFT_transform}) and simplifying yields the desired result.
\end{proof}
Based on this transform, we now proceed to derive a reformulation of the problem in (\ref{SPR_problem_BF}) which is amenable to iterative optimization.

\begin{theorem}\label{quadratic_SPR_equivalence}
	The optimization problem in (\ref{SPR_problem_BF}) is equivalent to the following auxiliary problem where, for notational clarity, we once again collect the auxiliary variables $\boldsymbol{\chi}_{k,b}$ in $\boldsymbol{X}$:
	\begin{subequations}\label{SPR_problem_quad}
		\begin{align}
		\underset{\boldsymbol{X},\mathbf{V^{\mathit{}}}}{\mathrm{maximize}}\quad f_{K_{q}}\left(\mathrm{vec}\left(\hat{\mathbf{R}}\left(\boldsymbol{X},\mathbf{V}\right)\right)\right)\hspace{5.85em}\label{SPR_obj_BF_transformed}\\
		\mathrm{subject\,to}\quad\sum_{k=1}^{K_{BS}}\left\Vert \mathbf{v}_{k,b}\right\Vert _{2}^{2}\leq P_{\mathrm{max}};\quad b=1,\ldots,B\label{sum_power_constraint_quad_BF}
		\end{align}
	\end{subequations}
	where $\hat{\mathbf{R}}\left(\boldsymbol{X},\mathbf{V}\right)$ is the $K\times{B}$ auxiliary objective matrix whose entries are given by
\[
\hat{\mathbf{R}}\left(\boldsymbol{X},\mathbf{V}\right)_{k,b}=\hat{r}_{k,b}\left(\boldsymbol{X},\mathbf{V}\right)
\]
	and $\hat{r}_{k,b}\left(\boldsymbol{X},\mathbf{V}\right)$ is given by (\ref{transformed_rates_QFT}).
	The equivalence between the two problems is in the sense of the optimal variables and objective function values.
\end{theorem}
\begin{proof}
	Observe that from Lemma \ref{quadratic_transform}, we have 
	\begin{equation}
	r_{k,b}\left(\mathbf{V}\right)=\underset{\boldsymbol{\chi}_{k,b}}{\mathrm{max}}\quad\hat{r}_{k,b}\left(\boldsymbol{X},\mathbf{V}\right)\label{QFT_rate_equivalence}
	\end{equation}
	where we identify 	
	\begin{equation}
	\mathbf{a}\left(\mathbf{V}\right)=\mathbf{H}_{b\rightarrow k,b}\mathbf{v}_{k,b}
	\end{equation}
	\begin{equation}
	\mathbf{B}\left(\mathbf{V}\right)=\sigma^{2}\mathbf{I}+\sum_{\left(k',b'\right)\neq\left(k,b\right)}\mathbf{H}_{b'\rightarrow k,b}\mathbf{v}_{k',b'}\mathbf{v}_{k',b'}^{\dagger}\mathbf{H}_{b'\rightarrow k,b}^{\dagger}
	\end{equation} 
	Hence it follows that
	\begin{subequations}
		\begin{align}
		f_{K_{q}}\left(\mathrm{vec}\left(\mathbf{R}\left(\mathbf{V}\right)\right)\right)=f_{K_{q}}\left(\underset{\boldsymbol{X}}{\mathrm{max}}\quad\mathrm{vec}\left(\hat{\mathbf{R}}\left(\mathbf{\boldsymbol{X}},\mathbf{V}\right)\right)\right)\label{quad_SLqP_equiv_1}\\
		=\underset{\boldsymbol{X}}{\mathrm{max}}\quad f_{K_{q}}\left(\mathrm{vec}\left(\hat{\mathbf{R}}\left(\mathbf{\boldsymbol{X}},\mathbf{V}\right)\right)\right)\label{quad_SLqP_equiv_2}
		\end{align}
	\end{subequations}
	where (\ref{quad_SLqP_equiv_1}) follows directly from (\ref{QFT_rate_equivalence}) and (\ref{quad_SLqP_equiv_2}) follows since the SLqP utility function is non-decreasing in each individual argument (according to Property 2 in Part I), allowing us to move the $\mathrm{max}\left(\cdot\right)$ operator outside $f_{K_{q}}\left(\cdot\right)$.
\end{proof}
	\begin{theorem}
		Problem (\ref{SPR_problem_quad}) is block-concave in $\boldsymbol{X}$ and $\mathbf{V}$.
	\end{theorem}			
\begin{proof}
	To verify the block-concavity of the objective, we begin by considering the beamforming variables $\mathbf{V}$ in the auxiliary rate function in (\ref{transformed_rates_QFT}). Observe that the term $2\mathrm{Re}\left\{ \mathbf{\boldsymbol{\chi}}_{k,b}^{\dagger}\mathbf{H}_{b\rightarrow k,b}\mathbf{v}_{k,b}\right\}$ is linear (and hence concave) in $\mathbf{v}_{k,b}$. Furthermore, note that $\mathbf{H}_{b'\rightarrow k,b}\mathbf{v}_{k',b'}\mathbf{v}_{k',b'}^{\dagger}\mathbf{H}_{b'\rightarrow k,b}^{\dagger}$ is a convex quadratic form. Hence it follows that 
	\[
	-\mathbf{\boldsymbol{\chi}}_{k,b}^{\dagger}\left(\sigma^{2}\mathbf{I}+\sum_{\left(k',b'\right)\neq\left(k,b\right)}\mathbf{H}_{b'\rightarrow k,b}\mathbf{v}_{k',b'}\mathbf{v}_{k',b'}^{\dagger}\mathbf{H}_{b'\rightarrow k,b}^{\dagger}\right)\mathbf{\boldsymbol{\chi}}_{k,b}
	\]
	is concave in $\mathbf{V}$ for fixed $\boldsymbol{\chi}_{k,b}$, since the interference-and-noise term is obviously positive semidefinite. Since $\mathrm{log}\left(\cdot\right)$ is a concave non-decreasing function, it follows that $\hat{r}_{k,b}\left(\boldsymbol{X},\mathbf{V}\right)$ is concave in $\mathbf{V}$ according to the standard rules of convex composition \cite{boyd_convex_2004}. Finally, we note that since $f_{K_{q}}\left(\cdot\right)$ is also a concave function and non-decreasing in each argument (according to Properties 2 and 1 respectively in Part I), it follows that the auxiliary objective $f_{K_{q}}\left(\mathrm{vec}\left(\hat{\mathbf{R}}\left(\mathbf{\boldsymbol{X}},\mathbf{V}\right)\right)\right)$ is concave in $\mathbf{V}$ when $\boldsymbol{X}$ is fixed according to the standard rules of convex composition \cite{boyd_convex_2004}.
	
	The concavity of $f_{K_{q}}\left(\mathrm{vec}\left(\hat{\mathbf{R}}\left(\mathbf{\boldsymbol{X}},\mathbf{V}\right)\right)\right)$ in $\boldsymbol{X}$ when $\mathbf{V}$ is fixed can be proven using a similar argument.
\end{proof}
	
The block-concavity of the auxiliary objective now enables us to derive an efficient algorithm to iteratively optimize the original objective. Observe that when $\mathbf{V}$ is fixed, an optimal choice of auxiliary variable $\boldsymbol{X}$ is given by (\ref{optimal_quadratic_variables}). Crucially, as with the power control problem, since the SLqP utility function is not strictly concave, there may be infinitely many optimal choices of $\boldsymbol{X}$ that optimize the auxiliary objective; however, the choice in (\ref{optimal_quadratic_variables}) maintains equivalence with the original objective as explained in Lemma \ref{quadratic_transform}. 

Conversely, when the auxiliary variables $\boldsymbol{X}$ are held fixed, we note that obtaining the optimal value of the beamforming variables $\mathbf{V}$ in (\ref{SPR_problem_quad}) is a convex optimization problem. Combining these two steps together, we obtain the MQFT algorithm for optimization of beamforming variables in Algorithm \ref{QFTAlg_BF}.

\begin{algorithm}[!h]
	\caption{Multidimensional Quadratic Fractional Transform Algorithm for SLqP Rate Maximization via Beamforming}
	\label{QFTAlg_BF}
	\begin{algorithmic}[1]
		\State \textbf{initialize} $\mathbf{V}$ such that it satisfies (\ref{sum_power_constraint_quad_BF}).
		\For{$i=1,\ldots$}
		\State \textbf{update} $\boldsymbol{X}$ using (\ref{optimal_quadratic_variables}).
		\State \textbf{update} $\mathbf{V}$ by solving (\ref{SPR_problem_quad}) for fixed $\boldsymbol{X}$.
		\EndFor
		\State \textbf{until} some convergence criterion is met.
	\end{algorithmic}
\end{algorithm}
\begin{remark}\label{QFT_nondecreasing_convergence}
	Algorithm \ref{QFTAlg_BF} is guaranteed to be non-decreasing in the original objective (\ref{SPR_obj_BF}) after each iteration.
\end{remark}
\begin{proof}
	The non-decreasing nature of the algorithm can be understood by following the chain of reasoning below, in which we denote the beamforming variables and auxiliary variables at the $i^\mathrm{th}$ iteration as $\mathbf{V}[i]$ and $\boldsymbol{X}[i]$ respectively. 
\begin{subequations}
		\begin{align}
		\hspace{-9em}\begin{array}{c}\hspace{-16.00em}
		f_{K_{q}}\left(\mathrm{vec}\left(\mathbf{R}\left(\mathbf{V}\left[i+1\right]\right)\right)\right)\\
		=\underset{\boldsymbol{X}}{\mathrm{max}}\quad f_{K_{q}}\left(\mathrm{vec}\left(\hat{\mathbf{R}}\left(\mathbf{\boldsymbol{X}},\mathbf{V}\left[i+1\right]\right)\right)\right)
		\end{array}\hspace{-2em}\label{convergence_proof_line_1}\\
		=f_{K_{q}}\left(\mathrm{vec}\left(\hat{\mathbf{R}}\left(\boldsymbol{X}\left[i+1\right],\mathbf{V}\left[i+1\right]\right)\right)\right)\hspace{-1.6em}\label{convergence_proof_line_2}\\
		\geq f_{K_{q}}\left(\mathrm{vec}\left(\hat{\mathbf{R}}\left(\boldsymbol{X}\left[i+1\right],\mathbf{V}\left[i\right]\right)\right)\right)\hspace{0.15em}\label{convergence_proof_line_3}\\
		\geq f_{K_{q}}\left(\mathrm{vec}\left(\hat{\mathbf{R}}\left(\boldsymbol{X}\left[i\right],\mathbf{V}\left[i\right]\right)\right)\right)\hspace{1.87em}\label{convergence_proof_line_4}\\
		=f_{K_{q}}\left(\mathrm{vec}\left(\mathbf{R}\left(\mathbf{V}\left[i\right]\right)\right)\right)\hspace{5.20em}\label{convergence_proof_line_5}
		\end{align}
\end{subequations}	
	
	The equality in (\ref{convergence_proof_line_1}) follows from (\ref{quad_SLqP_equiv_2}), while the subsequent inequality in (\ref{convergence_proof_line_2}) follows since $\boldsymbol{X}\left[i+1\right]$ optimizes the auxiliary objective while the beamforming variables are held fixed. The inequality in (\ref{convergence_proof_line_3}) follows since (\ref{SPR_obj_BF_transformed}) increases when it is optimized with respect to $\boldsymbol{X}$ while $\mathbf{V}$ is fixed. Likewise, the inequality in (\ref{convergence_proof_line_3}) follows since the auxiliary objective increases when it is optimized with respect to $\mathbf{V}$ while $\boldsymbol{X}$ is fixed. The inequality in (\ref{convergence_proof_line_4}) follows from similar reasoning to (\ref{convergence_proof_line_3}) as applied to the auxiliary variables $\boldsymbol{X}$. Finally, the equality in (\ref{convergence_proof_line_5}) follows from (\ref{quad_SLqP_equiv_1}) and (\ref{quad_SLqP_equiv_2}).
\end{proof}
\begin{remark}
	Algorithm \ref{QFTAlg_BF} is a cyclic minorization-maximization algorithm and is guaranteed to converge to a directional stationary point of Problem (\ref{SPR_problem_BF}).
\end{remark}
\begin{proof}
	The proof follows from similar reasoning to Theorem 4 in Part I.
\end{proof}

\begin{remark}
	The application of the quadratic fractional transform is distinct from that in \cite{shen_fractional_2018,shen_fractional_2018-1,khan_optimizing_2020}
\end{remark}
The QFT, as introduced in \cite{shen_fractional_2018}, was originally designed to apply to sum-of-functions-of-ratios (SOFOR) optimization problems. As such, it has primarily been utilized to optimize WSR, since the objective is in SOFOR form. However, the SLqP rate maximization problem is a pointwise minimum over a set of SOFORs. Therefore, the fractional programming approach in \cite{shen_fractional_2018,shen_fractional_2018-1} cannot be directly applied. 

It should be noted that the problem can technically be converted into SOFOR form by introducing slack variables. However, as explained earlier, this introduces a combinatorial number of constraints, which quickly becomes unmanageable. In contrast, the proposed MQFT approach directly tackles the non-smooth problem, thereby side-stepping this issue.

\subsection{Simulation Results}\label{simulation_model}
To illustrate the workings of Algorithm \ref{QFTAlg_BF}, we simulate a network consisting of $7$ hexagonal wrapped-around cells and BSs located at the center of each cell. The distance between adjacent BSs is set as $2000\hspace{0.1em}\mathrm{m}$, and we assume Rayleigh fading with a block-fading model. The path loss between a user and BS separated by a distance of $d$ meters is given by $\left(1+{d}/{d_{0}}\right)^{-\zeta/2}$ where $d_0=0.3920$ is a model-dependent reference distance, and $\zeta=3.76$ is the pathloss exponent. The noise power spectral density (PSD) is $-143\hspace{0.1em}\mathrm{dBm/Hz}$; we assume a system bandwith of $20\hspace{0.1em}\mathrm{MHz}$ and maximum per-user transmit power constraint of $43\hspace{0.1em}\mathrm{dBm}$. This path loss model is chosen in accordance with the COST-231 model which has been widely used in prior works to simulate typical LTE channel conditions \cite{khan_optimizing_2020,shen_fractional_2018-1}. Each base station is now equipped with $M_b=8$ transmit antennas, whereas each user is equipped with $N_{k_b}=2$ receive antennas. There are $K_{BS}=5$ users per cell, leading to a total of $K=35$ users in the network, and we choose $q=5.7$, i.e., we optimize the sum of the smallest $K_q=2$ rates in the entire network. 

We initialize the beamforming weights with an equal-power matched-filtering solution; in other words, the power for each user is set to $\frac{P_\mathrm{max}}{K_{BS}}$ similar to the schemes proposed in \cite{hosseini_optimizing_2018,khan_optimizing_2020}. Figure \ref{QFT_BF_convergence} illustrates the convergence of the network SLqP rate for a random set of channel realizations. There are two points of interest here. First, we note that the auxiliary objective exactly matches the original at each iteration; this is in line with the equivalence derived in Theorem \ref{quadratic_SPR_equivalence}. Second, the objective converges in a non-decreasing fashion as expected from Remark \ref{QFT_nondecreasing_convergence}.

\begin{figure}[t!]
	\begin{center} 
		\includegraphics[trim={0cm 0cm 0cm 0cm},clip,width=0.4\textwidth]{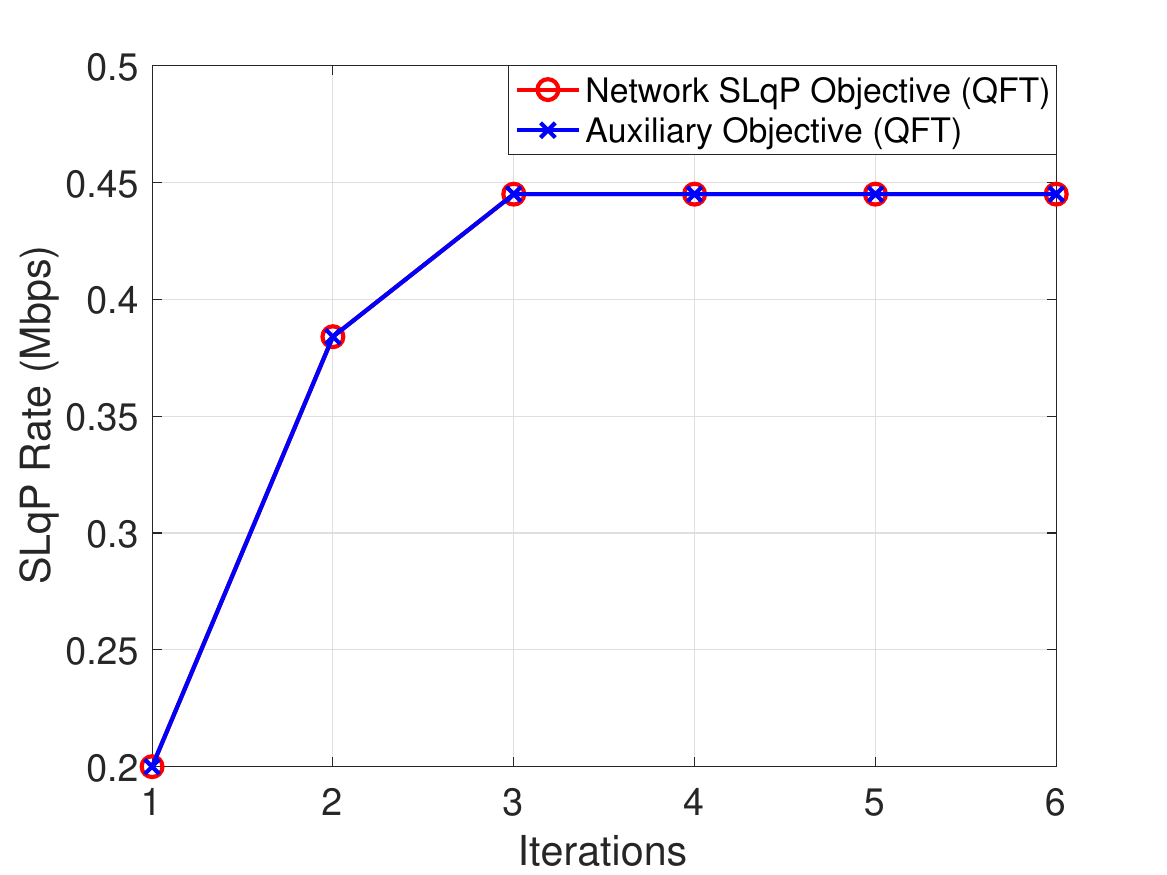}
		\caption{Convergence of SLqP rate objective for $K=35$, $K_{q}=2$.}
		\centering
		\label{QFT_BF_convergence}
	\end{center}
\end{figure} 

\section{Sum-Greatest $q^\mathrm{th}$ Weighted Mean Squared Error (SGqP-WMSE) Minimization}\label{SGqPsection}
\subsection{Background}

In designing an approach to the beamforming problem, it was natural to extend the LFT algorithm as applied to the power control problem. However, as we have seen, the problem in (\ref{SPR_problem_BF}) has its own unique characteristics. We now consider an alternative, simpler approach, to the beamforming problem. The WMMSE algorithm is widely utilized for a variety of resource allocation problems in wireless cellular networks. The relationship between MMSE and rate was first highlighted in \cite{christensen2008weighted} as $R=-\mathrm{log}\left(\mathrm{MSE}\right)$ and an iterative optimization algorithm was proposed which shown to be effective for the WSR maximization problem in wireless networks. Specifically, iterative minimization of a non-convex weighted MSE cost function is shown to effectively optimize network WSR; however, the convergence of the algorithm cannot be guaranteed.

The formulation subsequently proposed by Shi et al. in \cite{shi_iteratively_2011} is the most widely utilized since the associated algorithm is guaranteed to monotonically converge to a stationary point of the original WSR problem. Crucially, the work in \cite{shi_iteratively_2011} also demonstrated that the proposed problem formulation (and hence associated WMMSE algorithm) could be readily adapted for any \textit{smooth}, \textit{strictly concave} and \textit{sum-decomposable} utility function via its inverse gradient map \cite{shi_iteratively_2011}. This has led to its widespread adoption for a variety of resource allocation problems in cellular networks \cite{naghsh_maxmin_2019}. 

However, as we have already established, the SLqP utility functions, with the exception of the sum utility, are neither smooth nor strictly concave and are jointly computed across all users (and hence cannot be decomposed as a sum of per-user utilities). Smoothing introduces a combinatorial number of non-convex constraints, making optimization computationally intractable. Hence, the WMMSE approach cannot be used for SLqP rate optimization problems (other than $q=100$ or sum-rate).

To deal with these issues, we propose a different application of the classic WMMSE approach that enables us to \textit{directly tackle} the class of SLqP rate maximization problems, with specific reference to beamforming optimization in multiantenna networks. We first introduce a completely new problem, which we call the sum-greatest-$q^\mathrm{th}$-percentile weighted mean-squared (SGqP-WMSE) problem. While this might seem strange given that we are interested in SLqP-rate maximization, we demonstrate that there is, in fact, a deep connection between these two seemingly unrelated classes of problems. Specifically, we prove that solving the SGqP-WMSE problem is \textit{equivalent} to solving the problem of maximizing SLqP-rate. Thus, finding an algorithm to solve the former problem gives us another, independent, algorithm to maximize SLqP rate.

For readers who have skipped Part I, we would recommend going through Sections II-A and II-B in order to understand the properties of the SGqP utility function, which we utilize in the derivations that follow. 

\subsubsection{Proposed Approach}
Consider the downlink of a multiantenna network as given by the system model in Section \ref{QFT_BF_model_formulation}. Base station $b$ sends symbol $s_{k,b}\in\mathbb{C}$ to user $k$ within its cell, where we assume that $\mathbb{E}\left[\left|s_{k,b}\right|^{2}\right]=1\hspace{0.30em}\forall{\left(k,b\right)}$. This symbol is encoded using transmit beamformer $\mathbf{v}_{k,b}\in\mathbb{C}^{M_{b}\times1}$, i.e.,:
\[
\mathbf{x}_{k,b}=\mathbf{v}_{k,b}s_{k,b}
\]
\small
\begin{figure*}[!h]
	\footnotesize
	\hrulefill
	\begin{equation}\label{MSE}
	e_{k,b}\triangleq\mathbb{E}_{s,\mathbf{z}}\left[\left|\hat{s}_{k,b}-s_{kb}\right|^{2}\right]=\left|1-\mathbf{u}_{k,b}^{\dagger}\mathbf{H}_{b\rightarrow k,b}\mathbf{v}_{k,b}\right|^{2}+\sum_{\left(b',k'\right)\neq\left(b,k\right)}\mathbf{u}_{k,b}^{}\mathbf{H}_{b'\rightarrow k,b}\mathbf{v}_{k',b'}\mathbf{v}_{k',b'}^{\dagger}\mathbf{H}_{b\rightarrow k,b}^{\dagger}\mathbf{u}_{k,b}^{\dagger}+\sigma^{2}\mathbf{u}_{k,b}^{\dagger}\mathbf{u}_{k,b}
	\end{equation}
	\begin{equation}\label{optimal_receive_BF}
\mathbf{u}_{k,b}^{\mathrm{opt}}=\left(\sigma^{2}\mathbf{I}+\sum_{\left(b',k'\right)}\mathbf{H}_{b'\rightarrow k,b}\mathbf{v}_{k',b'}\mathbf{v}_{k',b'}^{\dagger}\mathbf{H}_{b'\rightarrow k,b}^{\dagger}\right)^{-1}\mathbf{H}_{b\rightarrow k,b}\mathbf{v}_{k,b}
\end{equation}
	\begin{equation}\label{MSE_one_variable}
	r_{k,b}\left(\mathbf{V}\right)=\underset{w_{k,b}}{\mathrm{min}}\quad w_{k,b}\left(\mathbf{I}-\mathbf{v}_{k,b}^{\dagger}\mathbf{H}_{b\rightarrow k,b}^{\dagger}\left(\sigma^{2}\mathbf{I}+\sum_{\left(b',k'\right)}\mathbf{H}_{b'\rightarrow k,b}\mathbf{v}_{k',b'}\mathbf{v}_{k',b'}^{\dagger}\mathbf{H}_{b'\rightarrow k,b}^{\dagger}\right)^{-1}\mathbf{H}_{b\rightarrow k,b}\mathbf{v}_{k,b}\right)-\mathrm{log}\left(w_{k,b}\right)-1
	\end{equation}
	\begin{equation}\label{optimal_wmmse_variable}
	w_{k,b}^{\mathrm{opt}}=\left(\mathbf{I}-\mathbf{v}_{k,b}^{\dagger}\mathbf{H}_{b\rightarrow k,b}^{\dagger}\left(\sigma^{2}\mathbf{I}+\sum_{\left(b',k'\right)}\mathbf{H}_{b'\rightarrow k,b}\mathbf{v}_{k',b'}\mathbf{v}_{k',b'}^{\dagger}\mathbf{H}_{b'\rightarrow k,b}^{\dagger}\right)^{-1}\mathbf{H}_{b\rightarrow k,b}\mathbf{v}_{k,b}\right)^{-1}
	\end{equation}
	\hrulefill
\end{figure*}
\normalsize
The received signal at user $k$ in cell $b$, denoted by $\mathbf{y}_{k,b}\in\mathbb{C}^{N_{k_{b}}\times1}$, is accordingly given by
\[
\mathbf{y}_{k,b}=\mathbf{H}_{b\rightarrow k,b}\mathbf{x}_{k,b}+\sum_{\left(k',b'\right)\neq\left(k,b\right)}\mathbf{H}_{b'\rightarrow k,b}\mathbf{x}_{k',b'}+\mathbf{z}_{k,b}
\]
where $\mathbf{z}_{k,b}\in\mathbb{C}^{N_{k_{b}}\times1}$ denotes the additive white Gaussian noise with distribution $\mathbf{z}_{k,b}\sim CN\left(\mathbf{0},\sigma^{2}\mathbf{I}\right)$. 

Assuming a linear receive beamformer $\mathbf{u}_{k,b}\in\mathbb{C}^{N_{k_{b}}\times1}$, the user in question then estimates the transmitted symbol $\hat{s}_{k,b}$ as
 \[\hat{s}_{k,b}=\mathbf{u}_{k,b}^{\dagger}\mathbf{y}_{k,b}\]

Assuming further that the symbols and receiver noise are statistically independent, the MSE is then given by (\ref{MSE}). 

Next, we define the weighted MSE (WMSE) function for user $k$ associated with BS $b$ as follows:
\begin{equation}\label{auxiliary_rate_wmmse}
\check{r}_{k,b}\left(\mathbf{W},\mathbf{U},\mathbf{V}\right)\triangleq w_{k,b}e_{kb}-\mathrm{log}\left(w_{k,b}\right)-1,
\end{equation}
where $e_{k,b}$ is the MSE given in (\ref{MSE}), and $w_{k,b}\in{\mathbb{R}_{++}}$ is an auxiliary optimization variable.
For notational convenience, we define the $K\times{B}$ weighted MSE matrix $\check{\mathbf{R}}\left(\mathbf{W},\mathbf{U},\mathbf{V}\right)$ to collect the WMSE values for the network under consideration, i.e.:
\[
\check{\mathbf{R}}\left(\mathbf{W},\mathbf{U},\mathbf{V}\right)_{k,b}=\check{r}_{k,b}\left(\mathbf{W},\mathbf{U},\mathbf{V}\right)
\]

Utilizing this weighted MSE matrix, we define the class of \textit{SGqP-WMSE minimization} problems below:
\begin{subequations}\label{SGqP_WMSE_problem}
	\begin{align}
	\underset{\mathbf{W},\mathbf{U},\mathbf{V}}{\mathrm{minimize}}\quad F_{K_{q}}\left(\mathrm{vec}\left(\check{\mathbf{R}}\left(\mathbf{W},\mathbf{U},\mathbf{V}\right)\right)\right)\hspace{4.70em}\label{SPR_obj_BF_transformed}\\
	\mathrm{subject\,to}\quad\sum_{k=1}^{K_{BS}}\left\Vert \mathbf{v}_{k,b}\right\Vert _{2}^{2}\leq P_{\mathrm{max}};\quad b=1,\ldots,B\label{sum_power_constraint_quad_BF}
	\end{align}
\end{subequations}
where, as in Part I, $F_{K_q}(\cdot)$ denotes the sum-greatest $q^\mathrm{th}$ percentile function, given by:
\[
F_{K_{q}}\left(\mathbf{x}\right)=\sum_{k=1}^{K_{q}}x_{i}^{\downarrow}
\]

We are now ready to state a crucial theorem that enables our optimization algorithm. 

\begin{theorem}\label{SGqP_WMSE_equivalence}
	The SGqP-WMSE minimization problem in (\ref{SGqP_WMSE_problem}) is equivalent to the SLqP-rate problem in (\ref{SPR_problem_BF}) in the sense that the optimal beamforming variable $\mathbf{V}$ and objective function values of both are identical.
\end{theorem}
\begin{proof}
	The equivalence can be proved as follows. First, we observe that $\check{r}_{k,b}\left(\mathbf{W},\mathbf{U},\mathbf{V}\right)$ is convex in $\mathbf{u}_{k,b}$ when $w_{k,b}$ is fixed and vice versa. Thus, setting the first-order condition of (\ref{MSE}) with respect to $\mathbf{u}_{k,b}$, we obtain the optimal (i.e., MMSE) receive beamformer as given in (\ref{optimal_receive_BF}). Substituting $\mathbf{u}_{k,b}^{\mathrm{opt}}$ back into (\ref{MSE}), it follows that the rate of the user in question can be equivalently expressed as in (\ref{MSE_one_variable}).
	Utilizing $\mathbf{u}_{k,b}^{\mathrm{opt}}$, we turn to the optimization of ${w}_{k,b}$. Considering the first-order condition of (\ref{MSE}) with respect to $w_{k,b}$, we obtain the optimal auxiliary variable as in (\ref{optimal_wmmse_variable}). Substituting $w_{k,b}^{\mathrm{opt}}$ back into (\ref{MSE_one_variable}), we recover the original rate expression in (\ref{data_rate_beamforming}) after some algebraic manipulation.
	
	From this, it follows that
	 \begin{multline}
	 \underset{\mathbf{W},\mathbf{U}}{\mathrm{min}}\quad F_{K_{q}}\left(\mathrm{vec}\left(\check{\mathbf{R}}\left(\mathbf{W},\mathbf{U},\mathbf{V}\right)\right)\right)=F_{K_{q}}\left(\mathrm{vec}\left(-\mathbf{R}\left(\mathbf{V}\right)\right)\right)\\
	 =F_{K_{q}}\left(-\mathrm{vec}\left(\mathbf{R}\left(\mathbf{V}\right)\right)\right)
	 \end{multline}
	
	Recall the symmetry property (i.e., Property 6) relating SLqP and SGqP functions introduced in Part I:
	\begin{equation}\label{SGqP_SLqP_symmetry}
	F_{K_{q}}\left(\mathbf{x}\right)=-f_{K_{q}}\left(\mathbf{-x}\right)
	\end{equation}
	
	Applying this property, it follows that 
	\[
	F_{K_{q}}\left(-\mathrm{vec}\left(\mathbf{R}\left(\mathbf{V}\right)\right)\right)=-f_{K_{q}}\left(\mathrm{vec}\left(\mathbf{R}\left(\mathbf{V}\right)\right)\right)
	\]
	
	The result then follows by noting that minimizing $-f_{K_{q}}\left(\mathrm{vec}\left(\mathbf{R}\left(\mathbf{V}\right)\right)\right)$ is equivalent to maximizing $f_{K_{q}}\left(\mathrm{vec}\left(\mathbf{R}\left(\mathbf{V}\right)\right)\right)$.
\end{proof}

\begin{remark}\label{sum_WMSE_max_WMSE_equivalence}
Problem (\ref{SGqP_WMSE_problem}) is equivalent to the sum-weighted MSE minimization problem when $K_q=1$ and the max-weighted MSE minimization problem when $K_q=K$.
\end{remark}
\begin{proof}
	This result follows straightforwardly from Property 3 of the SGqP utility function established in Part I.
\end{proof}
\subsubsection{Discussion} This equivalence between the SGqP-WMSE minimization and SLqP-rate maximization problems leads to a number of interesting consequences. The first, and most obvious one, is that Problem (\ref{SGqP_WMSE_problem}) subsumes the prior sum weighted-MSE and min-max WMSE problems introduced in \cite{shi_iteratively_2011} and \cite{razaviyayn_linear_2013} respectively as per Remark \ref{sum_WMSE_max_WMSE_equivalence}. This allows us to adopt a \textit{unified algorithmic approach} for optimizing SLqP rate problems rather than the ad hoc approach hitherto characteristic in dealing with special cases such as minimum-rate and sum-rate optimization. The key contribution of this algorithmic simplification is that it allows us to short-circuit the tedious and complexity-inducing steps advocated in prior works dealing with non-smooth problems such as hypograph reformulation (with a combinatorial number of constraints) \cite{razaviyayn_linear_2013,shen_fractional_2018-1} and successive convex approximation of non-convex constraints \cite{scutari_parallel_2017_1,scutari_parallel_2017}.

Significantly, as we shall explore later, this equivalence also allows us to adapt the WMSE transform to optimize the class of \textit{all} non-smooth utility functions that may not be strictly concave. This is in contrast to the original application of WMSE in \cite{shi_iteratively_2011} which, as mentioned earlier, could only be adapted to differentiable and strictly concave utility functions.

\begin{theorem}
	Problem (\ref{SGqP_WMSE_problem}) is convex in the transmit beamforming variables $\mathbf{V}$ when the receive beamforming variables $\mathbf{U}$ and weight variables $\mathbf{W}$ are fixed.
\end{theorem}
\begin{proof}
	Observe that (\ref{auxiliary_rate_wmmse}) is a linear function of $e_{k,b}$, which in turn is a convex function of $\mathbf{V}$ when $\mathbf{U}$ and $\mathbf{W}$ are fixed.
	Next, we note that the SGqP utility function is convex and non-decreasing in each argument as per Property 1 and 2 established in Part I. The result then follows from the standard rules of convex composition \cite{boyd_convex_2004}.
\end{proof}

The block-convexity of the SGqP-WMSE objective allows us to derive the iterative optimization approach given in Algorithm \ref{WMSEAlg_BF} to maximize SLqP rates.
\begin{algorithm}[!h]
	\caption{SGqP-WMSE Minimization Algorithm for SLqP Rate Maximization via Beamforming}
	\label{WMSEAlg_BF}
	\begin{algorithmic}[1]
		\State \textbf{initialize} $\mathbf{V}$ such that it satisfies (\ref{sum_power_constraint_quad_BF}).
		\For{$i=1,\ldots$}
		\State \textbf{update} $\mathbf{U}$ using (\ref{optimal_receive_BF}) for fixed $\mathbf{V,W}$.
		\State \textbf{update} $\mathbf{W}$ by using (\ref{optimal_wmmse_variable}) for fixed $\mathbf{U,V}$.
	    \State \textbf{update} $\mathbf{V}$ for fixed $\mathbf{U,W}$ by solving (\ref{SGqP_WMSE_problem}).
		\EndFor
		\State \textbf{until} some convergence criterion is met.
	\end{algorithmic}
\end{algorithm}
\begin{remark}\label{SGqP_obj_convergence}
Algorithm \ref{WMSEAlg_BF} is guaranteed to be non-decreasing in the original objective (\ref{SPR_obj_BF}) after each iteration.
\end{remark}
\begin{proof}
The non-decreasing convergence follows from similar reasoning to that in Remark \ref{QFT_nondecreasing_convergence}.
\end{proof}
\begin{remark}
	Algorithm \ref{WMSEAlg_BF} is a cyclic minorization-maximization algorithm, and is guaranteed to converge to a directional stationary point of Problem (\ref{SPR_problem_BF}).
\end{remark}
\begin{proof}
	The proof involves verifying the conditions detailed in Appendix C of Part I. To simplify the analysis, we consider maximization of the negative SGqP-WMSE objective; further, we denote the feasible set for problem (\ref{SPR_problem_BF}) as $\mathcal{V}$. We begin by noting that both the optimal receive beamformer and weight variable, given by (\ref{optimal_receive_BF}) and (\ref{optimal_wmmse_variable}) respectively, are functions of the beamforming variables $\mathbf{V}$ and thus subsequently denote them as $\mathbf{U}^{\mathrm{opt}}\left(\mathbf{V}\right)$ and $\mathbf{W}^{\mathrm{opt}}\left(\mathbf{V}\right)$ respectively.
	
	It follows that for all $\mathbf{V}\left[i\right]\in\mathcal{V}$, we have
	\begin{multline}
	-F_{K_{q}}\left(\mathrm{vec}\left(\check{\mathbf{R}}\left(\mathbf{W}^{\mathrm{opt}}\left(\mathbf{V}\left[i\right]\right),\mathbf{U}^{\mathrm{opt}}\left(\mathbf{V}\left[i\right]\right),\mathbf{V}\right)\right)\right)\\
	\leq-F_{K_{q}}\left(\mathrm{vec}\left(\check{\mathbf{R}}\left(\mathbf{W}^{\mathrm{opt}}\left(\mathbf{V}\right),\mathbf{U}^{\mathrm{opt}}\left(\mathbf{V}\right),\mathbf{V}\right)\right)\right)\hspace{6.200em}\\
	=f_{K_{q}}\left(\mathrm{vec}\left(\mathbf{R}\left(\mathbf{V}\right)\right)\right)\hspace{11.90em}
	\end{multline}
	 where the inequality follows since the choice of auxiliary variables $\mathbf{U}^\mathrm{opt}\left(\mathbf{V}\right)$ and $\mathbf{W}^\mathrm{opt}\left(\mathbf{V}\right)$ for a given value of beamforming variables $\mathbf{V}$ minimizes the SGqP-WMSE objective. The equality follows from similar reasoning to Theorem \ref{SGqP_WMSE_equivalence} and the application of the aforementioned symmetry property of the SLqP and SGqP functions. Hence, the negative SGqP-WMSE objective minorizes the original SLqP-rate objective in (\ref{SPR_obj_BF}), and Algorithm (\ref{WMSEAlg_BF}) is indeed an MM algorithm.
	 The proof of convergence to a directional stationary point follows from similar reasoning to Theorem 4 in Part I.
\end{proof}	

	Crucially, the identification of Algorithm \ref{WMSEAlg_BF} as belonging to the minorization-maximization framework greatly simplifies the analysis. The WMMSE algorithm in \cite{shi_iteratively_2011} analyzed the stationary points of the sum-WMSE minimization problem to show that they were equivalent to the stationary points of the beamforming sum-rate maximization problem. On the other hand, the work in \cite{razaviyayn_linear_2013} analyzed the dual of the epigraph form of the min-max WMSE problem, and relied upon identification of the inactive dual variables. This was then used in order to prove a one-to-one correspondence between the stationary points of the min-max WMSE problem and those of the original max-min-rate optimization problem. The direct application of this specific approach is not only extremely tedious for general SLqP-rate maximization problems (owing to the combinatorial number of non-convex constraints when the minimization problem is expressed in epigraph form). Thus, the proposed approach of characterizing the WMSE transform as an MM approach greatly reduces the complexity of analyzing the convergence properties of Algorithm \ref{WMSEAlg_BF}.
	\begin{remark}
	The WMSE transform can also be utilized for the SLqP rate maximization power control problems considered in Part I.
	\end{remark}
	A key point to consider here is that the WMSE transform introduces two auxiliary blocks of variables (i.e., $\mathbf{U}$ and $\mathbf{W}$); in contrast, the QFT and LFT transforms introduce just one additional block of variables (i.e., $\mathbf{X}$). When utilized for power control, therefore, the QFT and LFT transforms are more computationally efficient in the sense that they necessitate the update of two blocks of variables rather than the three associated with the WMSE transform.
	
\vspace{-1.2em}
\subsection{Simulation Results}
We simulate the performance of Algorithm \ref{WMSEAlg_BF} in an identical environment to that utilized for the QFT beamforming algorithm. There are $5$ users per cell, leading to a total of $35$ users in the network, and we optimize the sum of the smallest $K_q=2$ (corresponding to the $5.7^{\mathrm{th}}$ percentile) rates in the entire network. As with the QFT algorithm, we initialize the beamforming weights with an equal-power matched-filtering solution.

Figure \ref{SGqP_WMSE_BF_convergence} illustrates the convergence of the auxiliary SGqP-WMSE objective in (\ref{SGqP_WMSE_problem}) for a random set of channel realizations (chosen identically to those utilized for Figure \ref{QFT_BF_convergence}), while Figure \ref{SGqP_WMSE_BF_SLqP_rate_convergence} plots the convergence of the corresponding SLqP rate objective. As expected from Remark \ref{SGqP_obj_convergence}, the objective is non-decreasing in each iteration. The final SLqP utility achieved is nearly double the initial value. In comparison to the MQFT approach, the SGqP-WMSE approach achieves a slightly lower SLqP utility.

\begin{figure}
	\begin{subfigure}[b]{0.49\columnwidth}
		\includegraphics[trim={0cm 0cm 0cm 0cm},clip,width=1.0\columnwidth]{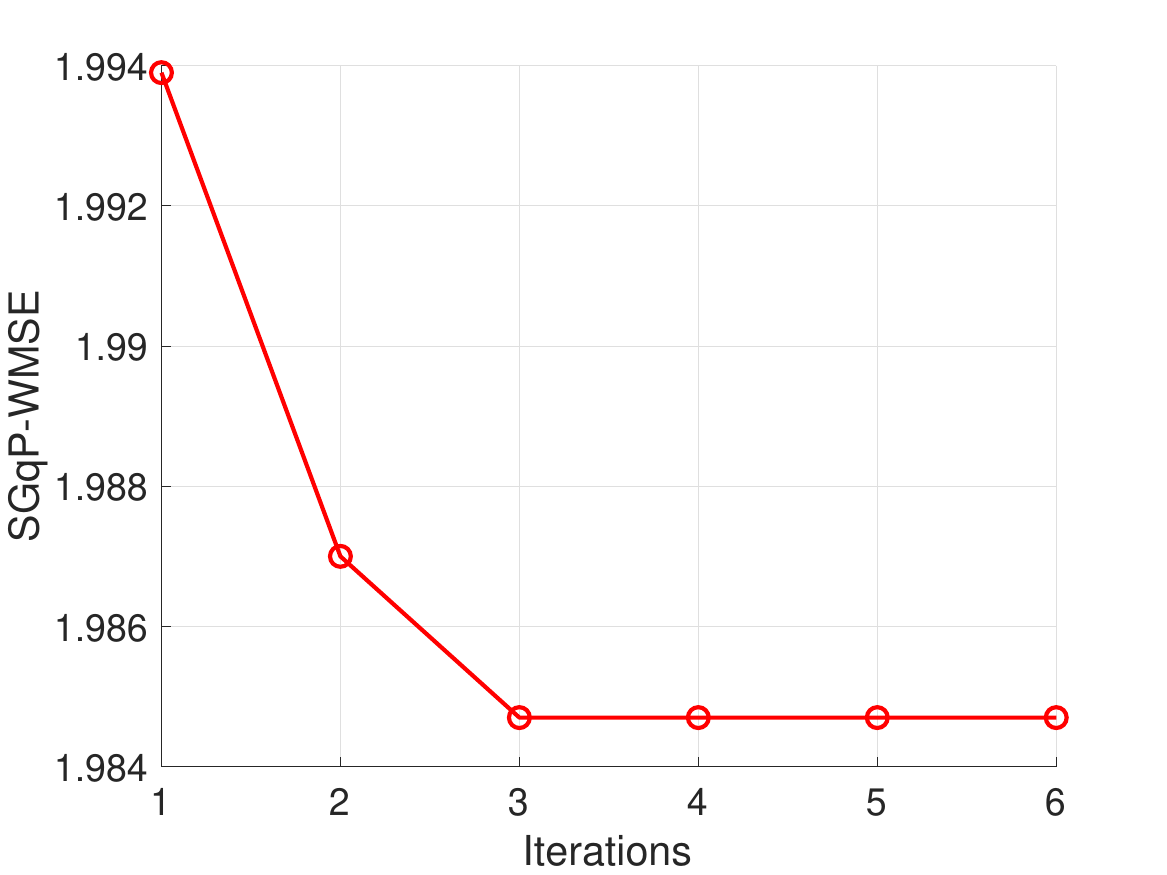}
		\caption{Convergence of SGqP-WMSE.}
		\label{SGqP_WMSE_BF_convergence}
	\end{subfigure}
	\hfill 
	\begin{subfigure}[b]{0.49\columnwidth}
		\includegraphics[trim={0cm 0cm 0cm 0cm},clip,width=1.0\columnwidth]{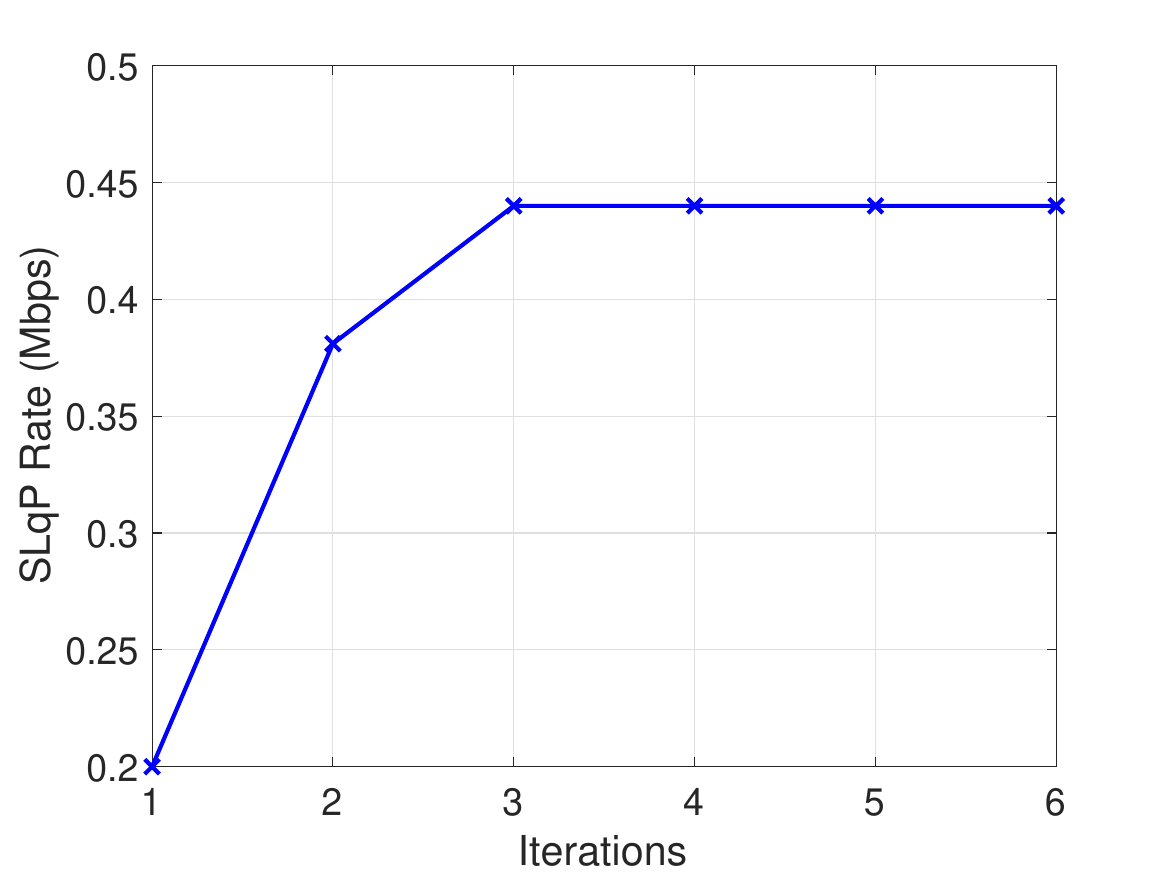}
		\caption{Convergence of SLqP rate utility.}
		\label{SGqP_WMSE_BF_SLqP_rate_convergence}
	\end{subfigure}
	\caption{Simultaneous convergence of SGqP-WMSE and SLqP-rate for $K=35$, $K_{q}=2$.}
\end{figure}

Finally, we consider an example with $K=14$ single-antenna users, $K_q=2$, and a base station equipped with $M=4$ transmit antennas, and compare against the following beamforming schemes from the prior literature:

\begin{itemize}
	\item Sequential Quadratic Programming (SQP): The beamforming variant of this approach is described in \cite{khan_optimizing_2020}; in essence, it iterates between finding the second-order approximation of the problem and solving the resulting approximated problem (which is a quadratic program). This algorithm has shown good performance in a variety of beamforming problems \cite{nguyen2017joint}. For our implementation, we utilize the built-in version of this algorithm in the fmincon package in MATLAB.
	
	\item Channel Weighted Sum-Rate (CWSR): Similar to the power control scheme, we also compare against the heuristic of solving a weighted sum-rate problem with the weights chosen inversely proportional to the channel strength. This should favour cell-edge users since they naturally have weaker information-bearing channels from their serving base stations. However, it cannot account for the cross-channel interference, and does not directly optimize the SLqP rate objective. We make use of the WMMSE algorithm \cite{shi_iteratively_2011}, which, as mentioned earlier, is guaranteed to reach a directional stationary point of the WSR objective.
	
	\item Zero-Forcing with Nulling (ZF-N): This approach is based on the beamforming scheme developed in \cite{hosseini_optimizing_2018}. In this scheme, the number of transmit antennas at the BS is assumed to be larger than the number of single-antenna users per cell. Accordingly, the BS can choose to distribute its spatial degrees of freedom (equal to the number of transmit antennas) in two different ways. It can use all of the degrees to design orthogonal beamforming weights for the users within its cell. Alternatively, it can use some (or all) of the excess degrees of freedom to design its beamforming weights so that they do not create interference for users being served in other cells. 
	
	This scheme has been shown to be effective in improving cell-edge rates; in particular, the $10^\mathrm{th}$-percentile rate is empirically shown to improve when the BS uses most of the excess spatial degrees of freedom to create nulls at these out-of-cell users. Accordingly, since each BS has $4$ degrees of freedom, and there are $2$ users per cell, we utilize the excess degrees of freedom to create nulls at the two out-of-cell users with the strongest channel strength.
\end{itemize}

The results are illustrated in Figure \ref{new_BF}. We observe that both the SGqP-WMSE and QFT algorithms converge smoothly, starting from the same initialization, to reach a directional stationary point. Similar to the power control setting, the CWSR scheme does not display monotonic convergence; this is to be expected since we are optimizing a different objective function (i.e., weighted sum-rate) rather than the actual one (i.e., SLqP rate). Furthermore, the SQP approach demonstrates poor performance, with almost no improvement compared to the initialization. In contrast, the ZF-N approach is better, but is still outperformed by our proposed solutions. 

We also note that the computation time involved with the beamforming schemes is, naturally, substantially higher as compared to the power control strategies.

\begin{figure}[t!]
	\begin{center} 
		\includegraphics[trim={0cm 0cm 0cm 0cm},clip, width=0.4\textwidth]{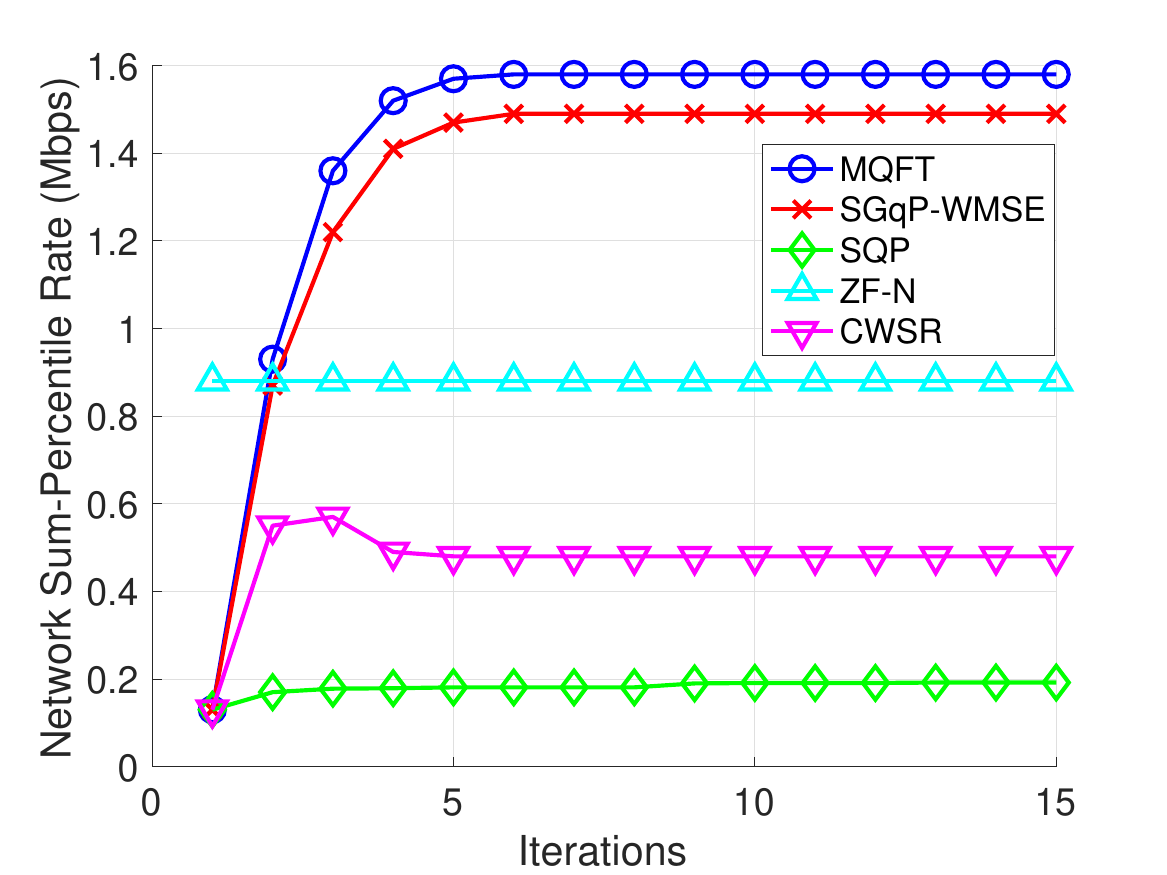}
		\caption{Convergence of SLqP utility for $K=14$ and $K_q=2$.}
		\centering
		\label{new_BF}
	\end{center}
\end{figure}

\section{Hybrid Utility Functions}\label{hybrids}

The key point to recognize is that the SLqP utility functions are not unique in their ability to be optimized in this fashion. In fact, the developed algorithms can be extended to \textit{any} concave utility function that is non-decreasing in each argument. This greatly expands the scope of the rate-based resource allocation problems we can solve and also allows us to re-interpret classical problems in the literature from a percentile perspective in order to derive effective solutions. Specifically, we can optimize new \textit{hybrid} utility functions, the construction of which we detail in the next section.

\subsection{Hybrid Utility Function Construction}\label{hybrid_construction}
Below, we list some of the most commonly optimized utility functions in wireless resource allocation; we refer to these as atomic utility functions since they form the basis from which we construct more complex utility functions later. We denote the rate achieved by user $k$ as $r_{k}$, dropping the notation indicating power or beamforming variables for clarity.
\begin{itemize}
	\item \textbf{SLqP utility functions:} $f_{K_{q}}\left(r_{1},\ldots,r_{K}\right)=\sum_{k=1}^{K_{q}}r_{k}^{\uparrow}$. As we have mentioned earlier, this family contains as special cases the max-min and sum-utility functions when $K_q=1$ and $K_q=K$ respectively.
	\item \textbf{Proportional fair utility:} $\mathrm{PF}\left(r_{1},\ldots,r_{K}\right)=\sum_{k=1}^{K}\mathrm{log}\left(r_{k}\right)$. As a consequence of the monotonicity of the logarithm function, this can also be viewed as maximization of the geometric mean of the rates, i.e.,:
	\small
	\[
	\mathrm{GM}\left(r_{1},\ldots,r_{K}\right)=\left(\prod_{k=1}^{K}r_{k}\right)^{1/K}
	\]
	\normalsize
	\item $\boldsymbol{\alpha}$-\textbf{fair utility}: For $\alpha>{0}$, $\alpha\neq{1}$, this family of utility functions is given by:
	\small
	\begin{equation}
	U_{\alpha}\left(r_{1},\ldots,r_{K}\right)=\sum_{k=1}^{K}\frac{1}{1-\alpha}r_{k}^{1-\alpha}
	\end{equation}
\normalsize
	When $\alpha=1$, this simply reduces to the proportional fair utility function; for $\alpha=0$ it is equivalent to the sum-utility function. Finally, in the limit as $\alpha\rightarrow\infty$, it is equivalent to the max-min utility function.
	\item \textbf{Harmonic mean:}  $\mathrm{HM}\left(r_{1},\ldots,r_{K}\right)=\frac{K}{\sum_{k=1}^{K}\frac{1}{r_{k}}}$.
\end{itemize}

Each of these utility functions is concave \cite{zhi-quan_luo_dynamic_2008} and non-decreasing in each argument. It follows that any utility function constructed through a concavity-preserving operation using any of these atomic utility functions with each other is therefore concave as well; in this paper, we focus on non-negative weighted sums and composition to construct \textit{hybrid} utility functions. Consequently, the QFT, LFT and SGqP-WMSE algorithms can all be directly utilized \textit{without modification} to optimize any such hybrid utility function; since the composition of concave non-decreasing functions is also concave non-decreasing, it follows that the update of the auxiliary variables would be unchanged from the setting of pure SLqP utility maximization.

To illustrate how these hybrid utility functions can be constructed and optimized, we consider three examples. For the sake of notational simplicity, we focus on the power control setting; however, the algorithms can be straightforwardly extended to the multiantenna beamforming system model using the approaches described in the previous section.
\vspace{-1.00em}
\subsection{Weighted Sum of SLqP Utilities}
So far, we have focused on pure SLqP utility optimization by targeting rate optimization at lower percentiles. In practice, this may be undesirable for operators who might wish to strike a balance between network throughput (which favours cell-centre users) and cell-edge rate. 

This balance can be achieved by optimizing a weighted sum of network throughput and $25^\mathrm{th}$-percentile rate as given by the problem below:
\begin{subequations}\label{SPR_problem_shortterm}
	\begin{align}
	\underset{\mathbf{p^{\mathit{}}}}{\mathrm{maximize}}\quad f_{K}\left(\mathbf{r}\left(\mathbf{p}\right)\right)+wf_{K/4}\left(\mathbf{r}\left(\mathbf{p}\right)\right)\hspace{5.05em}\label{SPR_obj_shortterm}\\
	\mathrm{subject\,to}\quad0\leq p_{k}^{}\leq P_{\mathrm{max}},\,k=1,\ldots,K\label{SPR_constraint_shortterm}\hspace{3.70em}
	\end{align}
\end{subequations}
where $w>0$ is a weight parameter controlling the tradeoff between sum-rate and $25^\mathrm{th}$-percentile SLqP rate. Clearly, when $w=0$, we are dealing with a pure sum-rate optimization problem and greater values of $w$ emphasize cell-edge rate. It should be emphasized that prior works have not considered solutions to these hybrid utility function optimization problems within the context of optimization theory. The work in \cite{naderializadeh_resource_2021} considers a deep reinforcement learning approach to optimize a similar hybrid function consisting of a weighted sum of network sum-rate and point-$5^\mathrm{th}$-percentile rate; however, the deep learning model needs to be retrained when the value of the weight parameter changes. In contrast, the proposed QFT/LFT algorithms can be straightforwardly applied and are \textit{guaranteed} to converge to directional stationary points regardless of the value of $w$ and do not need to be modified.

We plot the achieved sum-rate and $25^\mathrm{th}$-percentile SLqP rate for the chosen values of $w$ in Figure \ref{tradeoff}; the channel model and simulation settings are chosen identically to the $7$-cell network from Part I, with $K=56$ users in the network. As with Part I, and for all subsequent results in this paper, we initialize the power variables by choosing uniformly randomly in the interval $\left[0,P_\mathrm{max}\right]$.  For $w=0$ (i.e., the sum-rate), the cell-edge rate is $0$, as expected. As $w$ is increased to $10$, the cell-edge rate improves to approximately $6.1$ Mbps; for $w=100$, the cell-edge rate increases to $9.4$ Mbps. At the same time, the network sum-rate decreases substantially, from around $415$ Mbps when $w=0$ to $60$ Mbps when $w=100$. These results illustrate the fundamental tradeoff between favouring cell-edge and cell-centre users.

\begin{figure}[t!]
	\begin{center} 
		\includegraphics[trim={0cm 0cm 0cm 0cm},clip,width=0.4\textwidth]{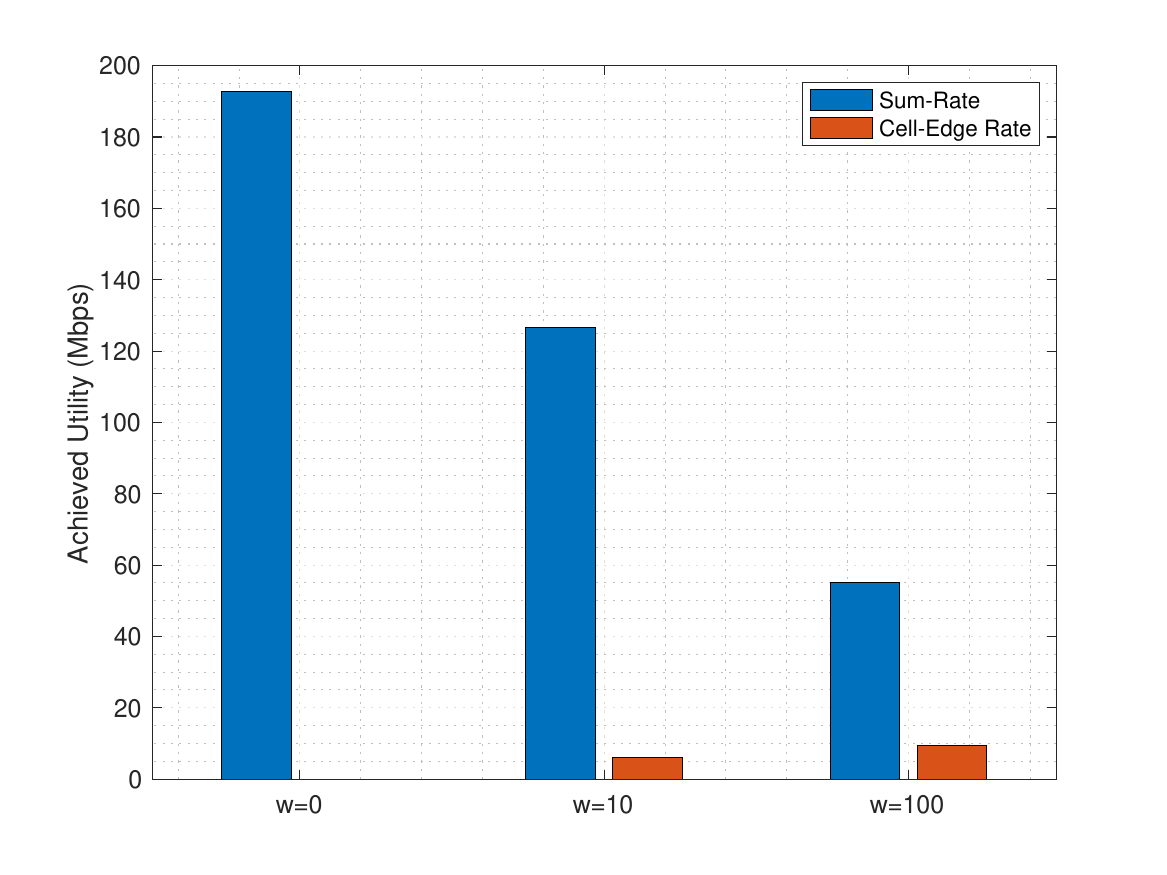}
		\caption{Tradeoff between sum-rate and cell-edge rate for different values of $w$.}
		\centering
		\label{tradeoff}
	\end{center}
\end{figure} 
\vspace{-1.00em}
\subsection{Per-Cell Percentile Utilities}

The problems we have considered so far involve optimization of \textit{network} SLqP utilities. However, it may be desirable to focus on some combination of \textit{per-cell} SLqP utilities, particularly in networks where a given cell may have a disproportionate share of underperforming users.

To elaborate, consider a network with $B$ cells, with BS $b$ serving $K_b$ users. The rates vector for the $b^\mathrm{th}$ cell is given by:
\[
\mathbf{r}_{b}\left(\mathbf{p}\right)=\left[r_{1,b}\left(\mathbf{p}\right),\ldots,r_{K_{b},b}\left(\mathbf{p}\right)\right]^{T}
\]

Note that unlike the scenarios we have considered so far, $K_b$ may differ from cell to cell. Thus, it may be the case that a given cell has a much larger number of users, leading to the network $10^\mathrm{th}$-percentile SLqP rate being largely controlled by users within this cell. Accordingly, BSs in all other cells may have to adjust their transmissions to avoid creating interference to users within the cell in question. This may once again be undesirable from an operator perspective.

A fairer approach can be devised as follows. Instead of optimizing the \textit{network} $10^\mathrm{th}$-percentile SLqP rate, we may choose to optimize the following functions of the \textit{per-cell} $10^\mathrm{th}$-percentile SLqP rates
\begin{itemize}
	\item {Arithmetic mean:} $\frac{1}{B}\sum_{b=1}^{B}f_{K_{b}/10}\left(\mathbf{r}_{b}\left(\mathbf{p}\right)\right)$.
	\item {{Geometric mean:}} $\left(\prod_{b=1}^{B}f_{K_{b}/10}\left(\mathbf{r}_{b}\left(\mathbf{p}\right)\right)\right)^{\frac{1}{B}}$.
	\item {Minimum of per-cell $10^\mathrm{th}$-percentile SLqP rate}: $\underset{b=1,\ldots,B}{\mathrm{min}}\quad f_{K_{b}/10}\left(\mathbf{r}_{b}\left(\mathbf{p}\right)\right)$.
\end{itemize}
subject to a per-user maximum power constraint identical to (\ref{SPR_constraint_shortterm}).
In contrast to the network $10^\mathrm{th}$-percentile SLqP rate, these hybrid utility functions ensure that every cell contributes to the overall objective.

Once again, each of these utility functions is concave overall following the discussion in Section \ref{hybrid_construction}. In particular, the minimum of per-cell $10^\mathrm{th}$-percentile SLqP rates can be viewed as the composition of the minimum-percentile SLqP utility function $f_{1/B}\left(\cdot\right)$ (which is concave and non-decreasing in each argument) with the per-cell $10^\mathrm{th}$-percentile SLqP utilities $f_{K_{b}/10}\left(\cdot\right)$ which are also concave, hence leading to overall concavity of the utility function.

Figure \ref{per_cell} plots the convergence of the QFT algorithm for each of these utility functions in a $7$-cell network with $K=140$ users. We note that the final values of the objective reached are ordered with the minimum per-cell $10^\mathrm{th}$-percentile rate smaller than the geometric mean which is in turn smaller than the geometric mean. This is to be expected in general according to the inequalities relating these general utilities established in \cite{zhi-quan_luo_dynamic_2008}.

We also consider these hybrid utilities for the $5^\mathrm{th}$-percentile utility; the results are illustrated in Figure \ref{5th_percentile_convergence}. The convergence and ordering follow a similar trend to the $10^\mathrm{th}$-percentile utilities given in Figure \ref{per_cell} (with the obvious difference that every $5^\mathrm{th}$-percentile utility function has a smaller value than the corresponding $10^\mathrm{th}$-percentile utility; this is to be expected from Remark 2 in Part I).

\begin{figure}
	\begin{subfigure}[b]{0.49\columnwidth}
		\includegraphics[trim={0cm 0cm 0cm 0cm},clip,width=1.0\textwidth]{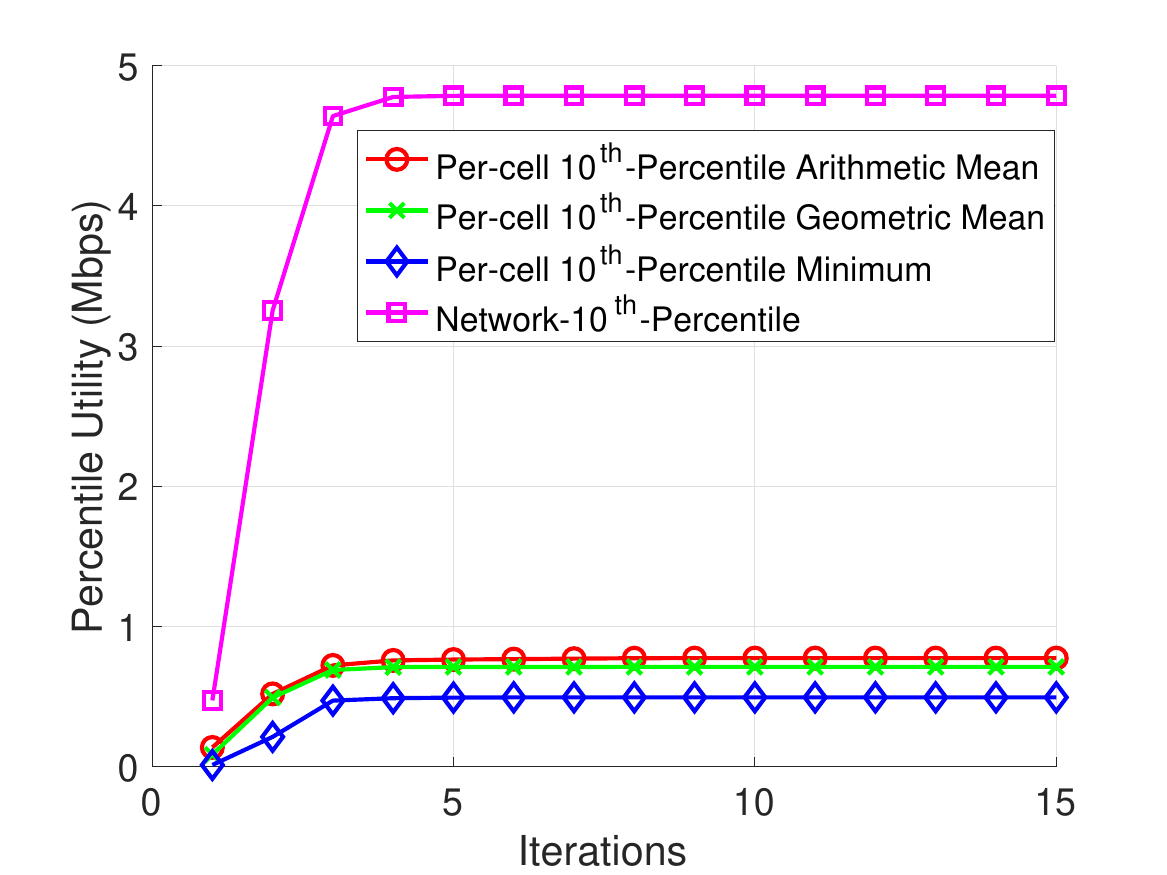}
		\caption{Hybrid $10^\mathrm{th}$-percentile objective convergence.}
	\centering
		\label{per_cell}
	\end{subfigure}
	\hfill 
	\begin{subfigure}[b]{0.49\columnwidth}
		\includegraphics[trim={0cm 0cm 0cm 0cm},clip, width=1.0\textwidth]{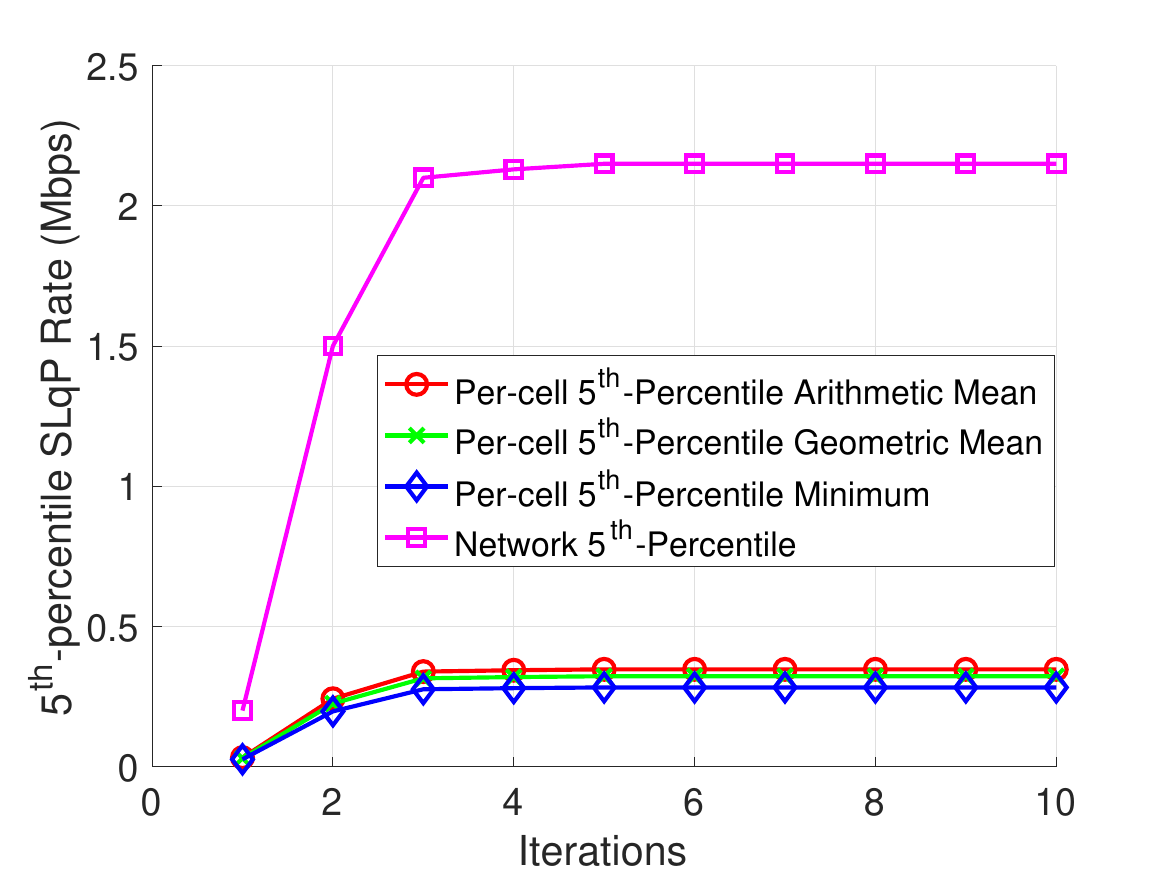}
		\caption{Hybrid $5^\mathrm{th}$-percentile objective convergence.}
	\centering
		\label{5th_percentile_convergence}
	\end{subfigure}
\caption{Convergence of QFT algorithm for hybrid $10^\mathrm{th}$-percentile and $5^\mathrm{th}$-percentile utilities}
\end{figure}
\vspace{-0.7em}
\subsection{Max-Min-Rate Across Multiple Frequency Bands} 
The max-min-rate problem can also be extended to the setting with multiple frequency bands. In this scenario, we wish to optimize the minimum (over all users) achieved sum-rate across multiple frequency bands. Specifically, consider the downlink of a multiuser SISO network with $K$ transmitter-receiver pairs and a total of $F$ orthogonal frequency bands. Denoting the channel for user $k$ on the $f^\mathrm{th}$ frequency band as $h_{f,k\rightarrow k}$ and the transmit power for the user and frequency band in question as $p_{f,k}$, it follows that the combined rate of user $k$ across all $F$ frequency bands is given by
\begin{equation}
R_{k}\left(\mathbf{P}\right)=\sum_{f=1}^{F}\mathrm{log}\left(1+\frac{p_{f,k}\left|h_{f,k\rightarrow k}\right|^{2}}{\sum_{j\neq k}p_{f,j}\left|h_{f,j\rightarrow k}\right|^{2}+\sigma_{f}^{2}}\right)
\end{equation}
where we collect the power variables for all frequency bands in $\mathbf{P}$ to avoid notational clutter, and $\sigma_{f}^{2}$ indicates the receiver noise power on frequency band $f$.

The overall optimization problem that describes our desired aim is then given by:
\small
\begin{subequations}\label{multi_freq_max_min}
	\begin{align}
	\underset{\mathbf{P^{\mathit{}}}}{\mathrm{maximize}}\quad\underset{k=1,\ldots,K}{\mathrm{min}}\quad R_{k}\left(\mathbf{P}\right)\hspace{6.95em}\label{mult_freq_obj}\\
	\mathrm{subject\,to}\quad p_{f,k}\geq0,\,\begin{array}{c}
	k=1,\ldots,K;\\
	f=1,\ldots,F
	\end{array}\label{mult_freq_constraint_1}\hspace{3.30em}\\
	\sum_{f=1}^{F}p_{f,k}\leq P_{\mathrm{max}},\,k=1,\ldots,K\label{mult_freq_constraint_2}\hspace{4.10em}
	\end{align}
\end{subequations}
\normalsize

Problem (\ref{multi_freq_max_min}) has been extensively studied in the literature and inherits the undesirable properties of both max-min utility maximization (i.e., non-smoothness) and sum-rate (i.e., non-convexity). Furthermore, it is known to be strongly NP-hard in the number of users $K$ \cite{zhi-quan_luo_dynamic_2008}. 

While smoothing this problem would result in $K$ non-convex constraints (which is manageable compared to general SLqP utility optimization problems as we have discussed), the strategies that we have developed allow us to simply bypass these intermediate steps. Observe that the ${\mathrm{min}_{{k=1,\ldots,K}}}\left(\cdot\right)$ operator can be viewed as the minimum percentile function $f_{1/K}(\cdot)$. Next, we recognize that the inner term $R_{k}\left(\mathbf{P}\right)$ is, in fact, a sum-utility function (and hence an SLqP utility). Since we are composing two concave utilities which are non-decreasing in each component, it follows that the overall utility is also concave. Hence the QFT and LFT algorithms can be directly applied to optimize the overall utility function to directional stationarity.

We plot the network utility achieved using the QFT algorithm as a function of transmit power in Figure \ref{mf_max_min}. These results were obtained by averaging over 1,000 realizations for each power level, with all algorithms initialized with the same uniform random power solution for fairness. We compare the QFT algorithm against the equal power allocation scheme (in which the total power for each user is equally divided among the $F=3$ orthogonal bands), and a heuristic of optimizing the sum-rate across all frequency bands and users while maintaining the power constraint in (\ref{mult_freq_constraint_2}). Additionally, we consider the following random power control schemes:

\begin{itemize}
	\item Uniform: For each user, we choose a transmission power for each of the three frequency bands uniformly randomly in the interval $\left[0,P_\mathrm{max}/3\right]$. This approach is standard and identical to that adopted by several prior resource allocation works for single-band optimization including \cite{nasir_multi-agent_2019,shen_fractional_2018}.
	\item Rayleigh: For each user, the power in each band is initialized from a Rayleigh distribution with parameter $P_\mathrm{max}\sqrt{1/2\pi}$, giving a mean value of $P_\mathrm{max}/2$.
	\item Exponential: For each user, the power in each band is initialized from an exponential distribution with parameter 1.
\end{itemize}

We observe that the sum-rate heuristic has extremely poor performance and is, in fact, outperformed by the equal-power as well as \textit{all} random power control schemes. This is unsurprising, since, as discussed earlier in Part I, sum-rate optimization is unfair to cell-edge users. Among the random power control schemes, the uniform random power control performs the best, although all three distributions achieve similar results. Nonetheless, the proposed QFT algorithm still achieves the best performance in terms of minimum rate.
\vspace{-0.30em}
\begin{figure}[!t]
	\begin{center} 
		\includegraphics[trim={0cm 0cm 0cm 0cm},clip, width=0.40\textwidth]{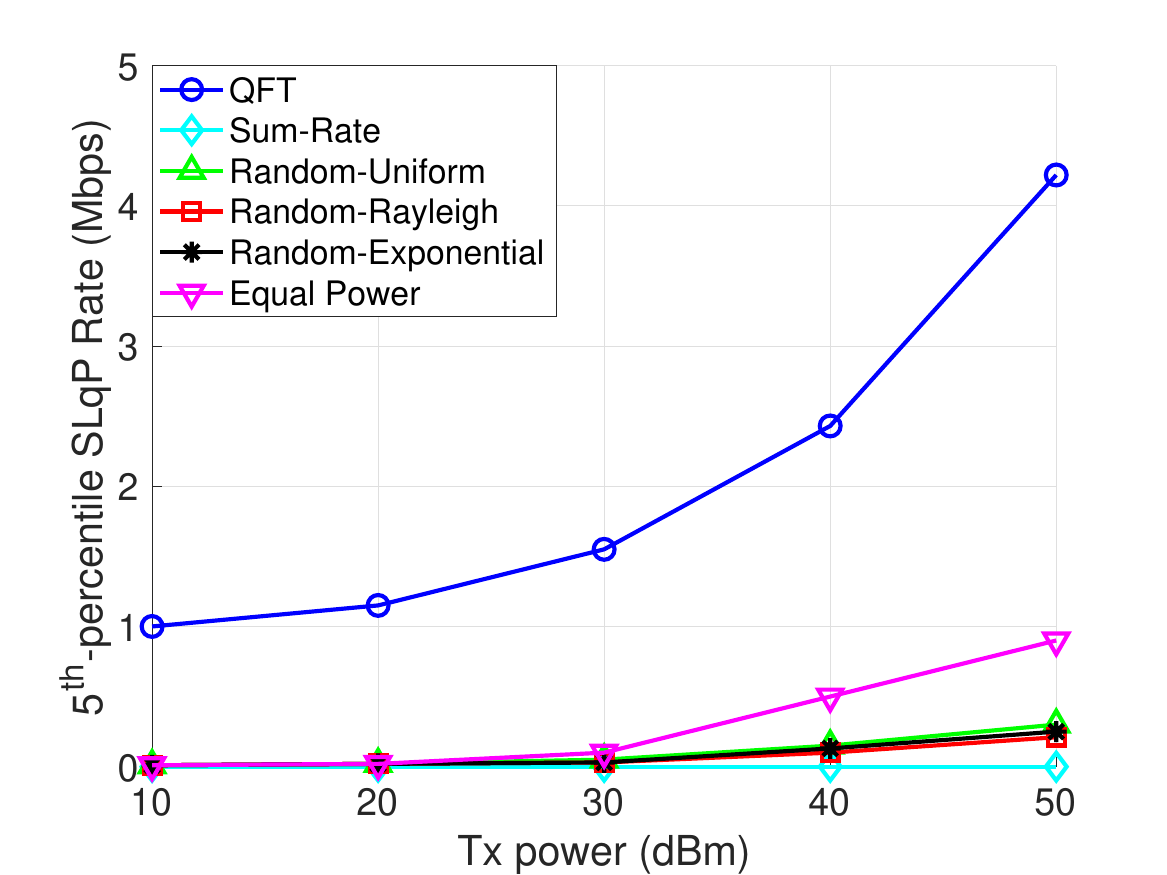}
		\caption{Minimum network sum-rate across $F=3$ frequency bands for a network with $K=56$ users as a function of $P_\mathrm{max}$.}
		\centering
		\label{mf_max_min}
	\end{center}
\end{figure}
\subsection{Proportional Fair SLqP Utility}
So far, we have focused on optimizing SLqP utility which corresponds to the sum of the smallest $K_q$ rates in the network. However, as with sum-rate, it is possible for certain users to be assigned zero powers, particularly if their channel conditions are poor. This is particularly undesirable in the context of improving service for cell-edge users. Furthermore, it raises the issue of how fairness can be guaranteed for the users experiencing the worst channel conditions.

One approach to ensure fairness is to choose to optimize the proportional fair SLqP rather than the sum-rate, as given by the problem formulation below:

\small
\begin{subequations}\label{SPR_problem_PF}
	\begin{align}
		\underset{\mathbf{p^{\mathit{}}}}{\mathrm{maximize}}\quad f_{K_{q}}\left(\mathrm{log}\left(r_{1}\left(\mathbf{p}\right)\right),\ldots,\mathrm{log}\left(r_{K}\left(\mathbf{p}^ {}\right)\right)\right)\hspace{2.85em}\label{SPR_PF}\\
		\mathrm{subject\,to}\quad0\leq p_{k}^{}\leq P_{\mathrm{max}},\,k=1,\ldots,K\label{}\hspace{5.00em}
	\end{align}
\end{subequations}
\normalsize

In other words, we would choose to maximize the \textit{sum of the logarithm} of the smallest $K_q$ rates. This formulation ensures fairness since, if any user achieves zero rate, the objective function value in (\ref{SPR_PF}) would become $-\infty$. On the other hand, any solution assigning non-zero power to each user would obviously achieve higher utility.

To illustrate this, we consider a numerical example for a network with $K=70$ users and $q=10$. Utilizing the QFT algorithm, we obtain the convergence plot for the objective function in (\ref{PF_SLqP}).
\begin{figure}[t!]
	\begin{center} 
		\includegraphics[trim={0cm 0cm 0cm 0cm},clip, width=0.4\textwidth]{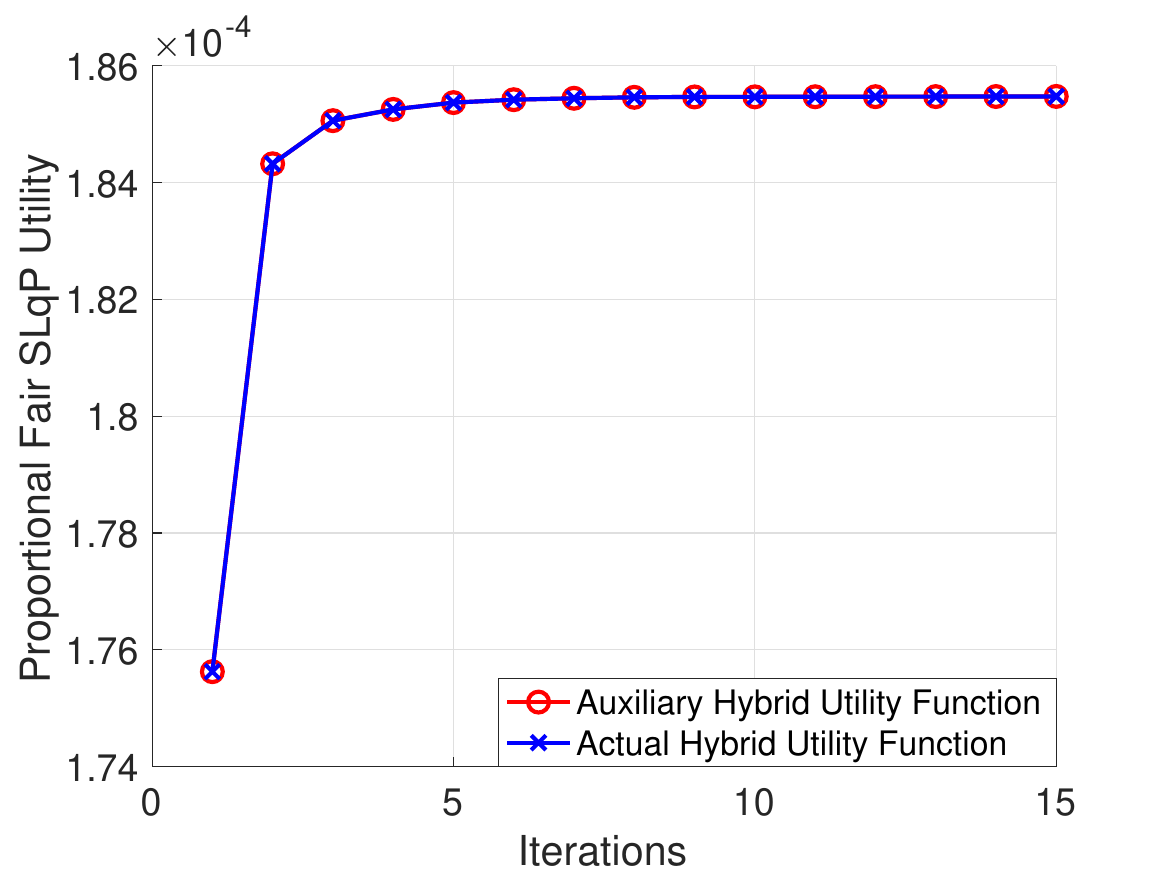}
		\caption{Convergence of proportional fair SLqP utility for $K=70$ and $K_q=7$.}
		\centering
		\label{PF_SLqP}
	\end{center}
\end{figure}

\subsection{Other Utility Functions}
Finally, we remark that the hybrid utility framework encompasses many more utility functions that have been studied in prior works. For example, the objective functions for the max-min multicasting problem \cite{karipidis_quality_2008} and weighted max-min-rate beamforming \cite{scutari_parallel_2017_1} problem can both be viewed as compositions of the minimum SLqP utility with non-negative weights and optimized using either the QFT or SGqP-WMSE framework.

\section{Long-Term Average Utility Maximization}\label{long_term}

So far we have focused on single-timeslot or \textit{short-term} percentile rate programs. In practice, however, it is more common to focus on ergodic or \textit{time-averaged utility} maximization. This results in a fairer resource allocation since it better accounts for both small- and large-scale fading effects in wireless networks, as well as factors like user mobility \cite{hossain_adaptive_2011}.
\vspace{-1.50em}
\subsection{Problem Formulation}
To illustrate this, we once again consider a multilink cellular network, with $K$ interfering single-antenna users; we note, however, that the subsequent approach we develop is equally applicable to the MU-MIMO beamforming model. The channel from the transmitter of user $j$ to the receiver of user $k$ in the $n^\mathrm{th}$ time slot is denoted by $h_{j\rightarrow k}^{\left(n\right)}$. Furthermore, the transmit power and AWGN receiver noise power for user $k$ are denoted by $p_{k}^{\left(n\right)}$ and ${z}_{k}^{\left(n\right)}$ respectively. It follows that the rate achievable by user $k$ for the $n^\mathrm{th}$ time slot is given by:
\begin{equation}\label{rates_expression}
\small
r_{k}^{\left(n\right)}\left(\mathbf{p}^{\left(n\right)}\right)=\mathrm{log}\left(1+\frac{p_{k}^{\left(n\right)}\left|h_{k\rightarrow k}^{\left(n\right)}\right|^{2}}{\sum_{j\neq k}p_{j}^{\left(n\right)}\left|h_{j\rightarrow k}^{\left(n\right)}\right|^{2}+{\sigma}^{2}}\right)
\end{equation}
\normalsize 

The long-term average or \textit{ergodic} utility of user $k$ is typically defined as a moving exponential weighted average of the rates achieved over all previous time slots as follows:
\begin{equation}\label{exponentially_weighted_average}
\bar{r}_{k}^{\left(n\right)}=\left(1-\alpha\right)\bar{r}_{k}^{\left(n-1\right)}+\alpha r_{k}^{\left(n\right)}
\end{equation}
where $\alpha\in(0,1)$ is a forgetting factor; we drop the notation indicating the power variables for brevity. We note, however, that our subsequent derivations are also compatible with other definitions of ergodic throughput, such as windowed averaging \cite{hossain_adaptive_2011}. Since direct maximization of the long-term utility is challenging, it is typically decomposed into a sequence of tractable sub-problems over each time slot. A well-known example of this is long-term average proportional-fair utility maximization in which a first-order WSR approximation is successively maximized in each time slot. Numerous works have explored algorithms for the general weighted sum-rate maximization problem; we refer the interested reader to  \cite{shen_fractional_2018-1} ,\cite{khan_optimizing_2020} and \cite{shi_iteratively_2011} as representative examples.

Unfortunately, such an approach is not possible with the SLqP utility: with the exception of the special case of sum-utility (i.e., $K_q=K$), its non-smoothness makes it impossible to derive a Taylor approximation. The challenge of optimization is further compounded by the fact that SLqP is a \textit{joint} utility function over all users; this coupling of utilities means that conventional long-term average utility maximization techniques \cite{palomar_tutorial_2006}, which decompose the network utility into a sum of per-user utility, cannot be readily applied either.

\subsection{Proposed Approaches}
To circumvent these obstacles, we propose to \textit{directly optimize} the long-term SLqP average utility instead of decomposing it. Assuming sequential optimization, it follows that the sub-problem for the $n^\mathrm{th}$ time slot can be expressed as

\small
\begin{subequations}\label{SPR_problem_longterm}
	\begin{align}
	\underset{\mathbf{p^{\mathit{\left(n\right)}}}}{\mathrm{maximize}}\quad f_{K_{q}}\left(\bar{r}_{1}^{\left(n\right)}\left(\mathbf{p}^{\left(n\right)}\right),\ldots,\bar{r}_{1}^{\left(n\right)}\left(\mathbf{p}^{\left(n\right)}\right)\right)\hspace{1.9em}\label{SPR_obj_longterm}\\
	\mathrm{subject\,to}\quad0\leq p_{k}^{\left(n\right)}\leq P_{\mathrm{max}},\,k=1,\ldots,K\label{SPR_constraint_longterm}\hspace{3.30em}
	\end{align}
\end{subequations}
\normalsize

where, from (\ref{exponentially_weighted_average}) and (\ref{rates_expression}), we have 
\small
\begin{multline}
\bar{r}_{k}^{\left(n\right)}\left(\mathbf{p}^{\left(n\right)}\right)=\left(1-\alpha\right)\bar{r}_{k}^{\left(n-1\right)}
\\+\alpha\mathrm{log}\left(1+\frac{p_{k}^{(n)}\left|h_{k\rightarrow k}^{(n)}\right|^{2}}{\sum_{j\neq k}p_{j}^{(n)}\left|h_{j\rightarrow k}^{(n)}\right|^{2}+\sigma^{2}}\right)
\end{multline}
\normalsize

The objective in (\ref{SPR_obj_longterm}) is once again non-convex and non-smooth. To effectively optimize it, we employ the QFT and LFT approaches to convert it into block-concave form.

We first consider the application of the QFT (given by Lemma 1 in Part I). Applying Theorem \ref{quadratic_transform} to Problem (\ref{SPR_problem_longterm}), we obtain the following equivalent auxiliary optimization problem:
\small
\begin{subequations}\label{SPR_problem_quad_longterm}
	\begin{align}
	\underset{\mathbf{x}^{\left(n\right)},\mathbf{p^{\mathit{\left(n\right)}}}}{\mathrm{maximize}}\quad f_{K_{q}}\left(\acute{r}_{1}^{\left(n\right)}\left(x_{1}^{\left(n\right)},\mathbf{p}^{\left(n\right)}\right),\ldots,\acute{r}_{K}^{\left(n\right)}\left(x_{K}^{\left(n\right)},\mathbf{p}^{\left(n\right)}\right)\right)\hspace{0.em}\label{SPR_obj_quad_longterm}\\
	\mathrm{subject\,to}\quad0\leq p_{k}^{\left(n\right)}\leq P_{\mathrm{max}},\,k=1,\ldots,K\label{SPR_constraint_longterm_quad}\hspace{5.80em}
	\end{align}
\end{subequations}
\normalsize
where
\small
\begin{multline}
\acute{r}_{k}^{\left(n\right)}\left(\mathbf{p}^{\left(n\right)}\right)=\left(1-\alpha\right)\bar{r}_{k}^{\left(n-1\right)}\\+\alpha\mathrm{log}\left(1+2x_{k}^{\left(n\right)}\sqrt{A_{k}\left(\mathbf{p}^{\left(n\right)}\right)}-\left(x_{k}^{\left(n\right)}\right)^{2}B_{k}\left(\mathbf{p}^{\left(n\right)}\right)\right)
\end{multline}
\normalsize
and $A_{k}\left(\mathbf{p}^{\left(n\right)}\right)$ and $B_{k}\left(\mathbf{p}^{\left(n\right)}\right)$ are once again the signal power and interference-plus-noise power in the SINR respectively.

The application of the LFT is similar. Utilizing Lemma 1 from Part I for the objective in (\ref{SPR_obj_longterm}), we obtain the following equivalent auxiliary optimization problem:

\small
\begin{subequations}\label{SPR_problem_log_longterm}
	\begin{align}
	\underset{\mathbf{x}^{\left(n\right)},\mathbf{p^{\mathit{\left(n\right)}}}}{\mathrm{maximize}}\quad f_{K_{q}}\left(\grave{r}_{1}^{\left(n\right)}\left(x_{1}^{\left(n\right)},\mathbf{p}^{\left(n\right)}\right),\ldots,\grave{r}_{K}^{\left(n\right)}\left(x_{K}^{\left(n\right)},\mathbf{p}^{\left(n\right)}\right)\right)\hspace{0.em}\label{SPR_obj_log_longterm}\\
	\mathrm{subject\,to}\quad0\leq p_{k}^{\left(n\right)}\leq P_{\mathrm{max}},\,k=1,\ldots,K\label{SPR_constraint_longterm_log}\hspace{5.80em}
	\end{align}
\end{subequations}
\normalsize

where
\small
\begin{multline}
\grave{r}_{k}^{\left(n\right)}\left(\mathbf{p}^{\left(n\right)}\right)=\left(1-\alpha\right)\bar{r}_{k}^{\left(n-1\right)}-\alpha x_{k}^{\left(n\right)}B_{k}\left(\mathbf{p}^{\left(n\right)}\right)+\\
\mathrm{\alpha log}\left(x_{k}^{\left(n\right)}\left(A_{k}\left(\mathbf{p}^{\left(n\right)}\right)+A_{k}\left(\mathbf{p}^{\left(n\right)}\right)\right)\right)+\alpha
\end{multline}
\normalsize

Following similar reasoning to Remark 8 from Part I, it can be easily verified that the auxiliary objectives in (\ref{SPR_obj_log_longterm}) and (\ref{SPR_problem_quad_longterm}) are block-concave in the power variables when the auxiliary variables are held fixed. The updates for the auxiliary variables remain unchanged compared to the short-term optimization, and are given respectively for the QFT and LFT by (\ref{optimal_quadratic_variables_power}) and (\ref{optimal_logarithmic_variables_power}) below.

\small
\begin{subequations}
	\begin{align}
	\mathrm{\mathbf{QFT:}}\hspace{2.00em}x_{k}^{\left(n\right)}=\frac{\sqrt{A_{k}\left(\mathbf{p}^{\left(n\right)}\right)}}{B_{k}\left(\mathbf{p}^{\left(n\right)}\right)},\quad{k}{=}{1,\ldots,K}\label{optimal_quadratic_variables_power}\\
	\mathrm{\mathbf{LFT:}}\hspace{2.430em}x_{k}^{\left(n\right)}=\frac{1}{B_{k}\left(\mathbf{p}^{\left(n\right)}\right)},\quad{k}=1,\ldots,K\label{optimal_logarithmic_variables_power}
	\end{align}
\end{subequations}
\normalsize
The update for the power variables then involves solving a convex optimization problem with fixed auxiliary variables, which is straightforward. 

Combining all steps together, we obtain Algorithm \ref{QFTAlg_longterm} for iterative optimization of the long-term average rate using the QFT and LFT algorithms. As with their short-term counterparts, the proposed algorithms are guaranteed to converge to stationary points of the per time slot problem in (\ref{SPR_problem_longterm}).

\begin{algorithm}[!t]
	\caption{QFT/LFT Algorithm for Long-Term SLqP Rate Maximization}
	\label{QFTAlg_longterm}
	\begin{algorithmic}[1]
		\For{$n=1,\ldots$}
		\State \textbf{initialize} $\mathbf{p}^{\left(n\right)}$
		\For{$i=1,\ldots$}
		\State{QFT:} \textbf{update} $\mathbf{x}^{\left(n\right)}$ using (\ref{optimal_quadratic_variables_power}).
		\State \hspace{2.50em}\textbf{update} $\mathbf{p}^{\left(n\right)}$ by solving (\ref{SPR_problem_quad_longterm}) for fixed $\mathbf{\hspace{10.50em}x}^{\left(n\right)}$.
		\hspace{2.00em}\State{LFT:} \textbf{update} $\mathbf{x}^{\left(n\right)}$ using (\ref{optimal_logarithmic_variables_power}).
		\State \hspace{2.50em}\textbf{update} $\mathbf{p}^{\left(n\right)}$ by solving (\ref{SPR_problem_log_longterm}) for fixed $\mathbf{\hspace{10.50em}x}^{\left(n\right)}$.
		\EndFor
		\State \textbf{until} some convergence criterion is met.
		\State \textbf{update} $\bar{r}_{k}^{\left(n\right)}$ for $k=1,\ldots,K$ using (\ref{exponentially_weighted_average}).
		\EndFor
	\end{algorithmic}
\end{algorithm}
\vspace{-0.0em}

\subsection{Simulation Results}
Utilizing a simulation environment identical to that detailed in Section \ref{simulation_model}, we plot the cumulative distribution function of long-term average $20^\mathrm{th}$-percentile throughput for both the QFT and LFT algorithms in Figure \ref{long_term_average_results}. The forgetting factor for the averaging process was chosen as $\alpha=0.30$. As benchmarks, we compare against the WMMSE algorithm for proportionally fair WSR maximization. 

In this multiple time slot setting, the proportionally fair WSR maximization problem is \textit{not} convex; indeed, as prior works have noted, the problem is strongly NP-hard in general. Thus, the WMMSE algorithm is guaranteed, at best, to converge to a stationary point of the long-term average sum-log-utility problem. The results are illustrated in Figure \ref{long_term_average_results}. 

As we can observe, both the proposed approaches achieve significantly higher long-term average utility than the WMMSE algorithm. The averaged utility obtained is summarized in Table \ref{long_term_average_results_summary}; both the QFT and LFT algorithms achieve over $50\%$ higher long-term average utility as compared to the WMMSE algorithm. However, an interesting point to note is that the LFT algorithm slightly outperforms the QFT algorithm; in our experience, this generally holds true for long-term utility maximization at lower percentiles. In contrast, the QFT algorithm appears to perform better for the short-term SLqP rate maximization problem as illustrated by the prior results in Part I. It is worth emphasizing that this is an observation and not a theoretical result.

Finally, it can be seen that the execution time of the proposed schemes is longer than the WMMSE-PF benchmark. This is largely due to the fact that WMMSE is a closed-form algorithm, while, as mentioned in Part 1, we call CVX to solve a convex problem in each iteration of the QFT and LFT algorithms. While this is a drawback, it is worth re-emphasizing that the performance gap between our schemes and the WMMSE-PF benchmark is quite large.
\begin{figure}[t!]
	\begin{center} 
		\includegraphics[trim={0cm 0cm 0cm 0cm},clip,width=0.4\textwidth]{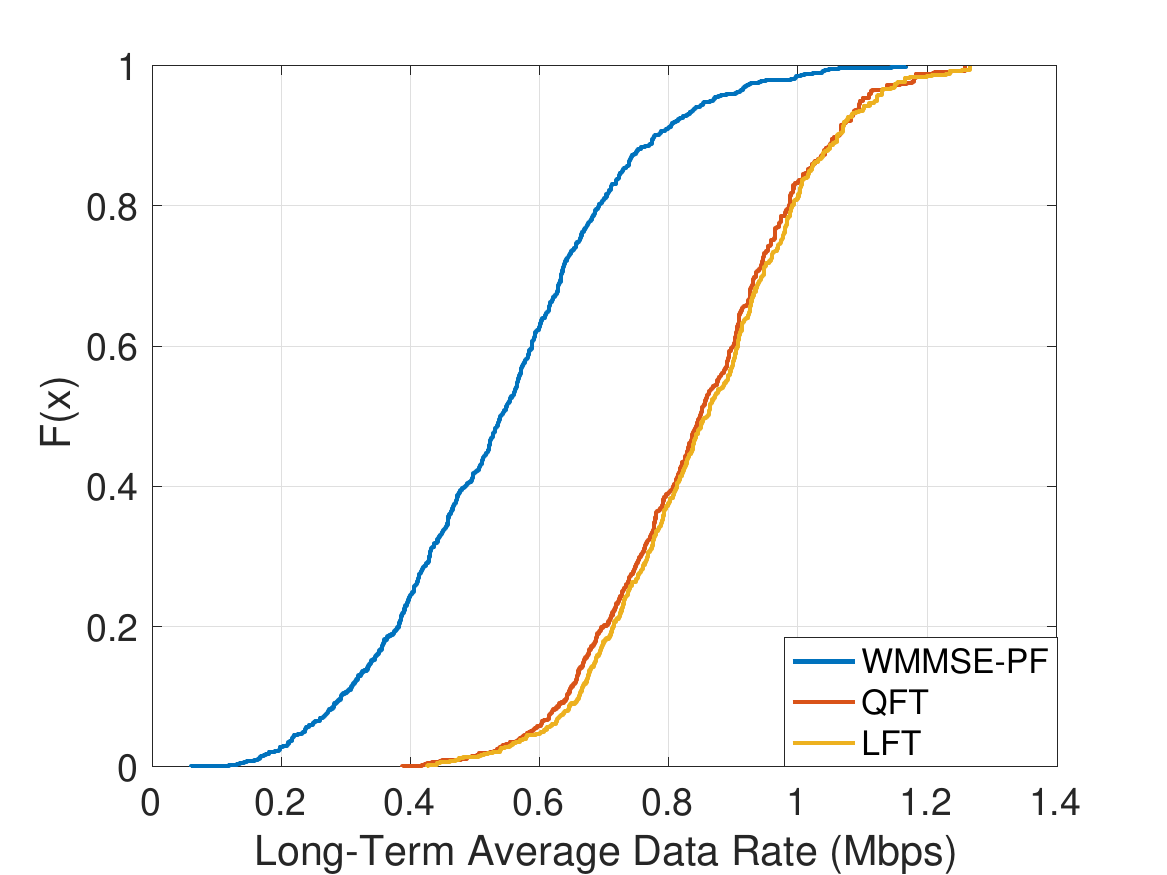}
		\caption{Empirical distribution of long-term average SLqP rates for $K=140$, $K_q=28$.}
		\centering
		\label{long_term_average_results}
	\end{center}
\end{figure} 
\vspace{-0.600em}
\begin{table}[h!]

	\begin{center} 
		\captionof{table}{Long-term average $20^\mathrm{th}$-percentile SLqP rate for different algorithms.} \label{tab1}
		\begin{tabular}{|c|c|c|}
			\hline
			Algorithm & Ergodic Utility (Mbps) & Run Time (s)\\ \hline
			QFT       & 11.8501  &         65        \\ \hline
			LFT       & 11.9876     &        71      \\ \hline
			WMMSE-PF  & 7.5604 &           5         \\ \hline
		\end{tabular}
		\centering
		\label{long_term_average_results_summary}
	\end{center} 
\end{table}
\section{Conclusion}\label{conklusion}
This paper explores the problem of optimizing beamforming weights in MU-MIMO networks for SLqP utility functions. As in Part I, we are interested in optimizing resources to help cell-edge users. However, as we have seen, formulation can be used for a wide variety of objective functions. We proposed a multidimensional complex-valued extension to the QFT algorithm as well as the SGqP-WMSE formulation which enable us to derive effective non-decreasing MM algorithms for this class of problems. Our proposed algorithms also outperform benchmarks from the prior beamforming literature, including the ZF-N and SQP methods. Since these algorithms belong to the minorization-maximization framework, they are guaranteed to converge to directional stationary points of the rate-percentile utility functions and enable us to short-circuit the tedious steps conventionally employed in dealing with non-smooth and non-convex problems.

We also demonstrated how the proposed algorithms can be directly applied to optimize \textit{any concave, non-decreasing utility function}. This enables us to extend optimization to the class of hybrid utility functions, which subsume numerous resource allocation problems considered in prior works. At the same time, this also enables us to combine percentile and conventional utility functions in order to guarantee fairness for all users.

Finally, we considered long-term utility maximization, and proposed a direct optimization strategy that avoids the differentiability requirement imposed by classic approximation techniques. The proposed QFT and LFT algorithms considerably outperform the benchmark proportional-fair WMMSE algorithm in terms of lower-percentile cell-edge throughput.

This paper represents a preliminary step towards understanding and solving this important class of problems, and there are numerous directions which can be pursued in the future. One pressing question is how to reduce the complexity associated with solving convex problems for each iteration of the proposed algorithms. Finally, incorporating further constraints on individual throughputs in conjunction with the proposed hybrid utilities, may allow for even greater flexibility in terms of controlling the allocation of resources to different users in wireless networks.



\bibliographystyle{IEEEtran}
\bibliography{IEEEabrv.bib,biblio.bib}

\end{document}